\documentclass{scrartcl}

\usepackage{fullpage}

\usepackage{amsmath,amsthm,amssymb}
\usepackage{xspace}

\usepackage[ruled,vlined]{algorithm2e}

\usepackage{tikz}

\usepackage{booktabs,multirow,array}

\newcommand{\N}{\mathbb{N}}

\newcommand{\Oh}{\mathcal{O}}

\newcommand{\kparam}[1]{\ensuremath{k_\mathrm{#1}}\xspace}

\newcommand{\klies}{\kparam{lies}}
\newcommand{\kfaults}{\kparam{faults}}
\newcommand{\ksum}{\kparam{sum}}
\newcommand{\kmax}{\kparam{max}}
\newcommand{\krep}{\kparam{rep}}
\newcommand{\kbswap}{\kparam{bswap}}
\newcommand{\krbswap}{\kparam{rbswap}}
\newcommand{\kbmov}{\kparam{bmov}}
\newcommand{\kseq}{\kparam{seq}}
\newcommand{\kswap}{\kparam{swap}}
\newcommand{\kaswap}{\kparam{aswap}}
\newcommand{\kinv}{\kparam{inv}}
\newcommand{\kainv}{\kparam{ainv}}
\newcommand{\kmov}{\kparam{mov}}

\theoremstyle{plain}
\newtheorem{lemma}{Lemma}
\newtheorem{theorem}{Theorem}
\newtheorem{proposition}{Proposition}
\newtheorem{corollary}{Corollary}
\newtheorem{claim}{Claim}
\theoremstyle{definition}

\theoremstyle{remark}

\DeclareMathOperator{\rank}{rank}
\DeclareMathOperator{\pos}{pos}
\DeclareMathOperator{\prev}{prev}
\DeclareMathOperator{\nxt}{next}
\DeclareMathOperator{\scc}{succ}
\DeclareMathOperator{\query}{query}

\makeatletter
\def\@xfootnote[#1]{%
  \protected@xdef\@thefnmark{#1}%
  \@footnotemark\@footnotetext}
\makeatother
\makeatletter
\newcommand\footnoteref[1]{\protected@xdef\@thefnmark{\ref{#1}}\@footnotemark}
\makeatother

\title{Robust and adaptive search}
\author{Yann Disser\footnote{Main work done while at TU Berlin. Current address: TU Darmstadt, \texttt{disser@mathematik.tu-darmstadt.de}} \and Stefan Kratsch\footnote{University of Bonn, \texttt{kratsch@cs.uni-bonn.de}}}

\begin{document}

\maketitle

\begin{abstract}
Binary search finds a given element in a sorted array with an optimal number of~$\log n$ queries. 
However, binary search fails even when the array is only slightly disordered or access to its elements is subject to errors. 
We study the worst-case query complexity of search algorithms that are robust to imprecise queries and that adapt to perturbations of the order of the elements.
We give (almost) tight results for various parameters that quantify query errors and that measure array disorder.
In particular, we exhibit settings where query complexities of~$\log n + ck$, $(1+\varepsilon)\log n + ck$, and $\sqrt{cnk}+o(nk)$ are best-possible for parameter value~$k$, any~$\varepsilon>0$, and constant~$c$.  
\end{abstract}


\section{Introduction}

Imagine a large register with $n$ files from which you wish to extract a particular file. 
All files are indexed by some key and the files are sorted by key value. 
Not knowing the distribution of the keys, you probably use binary search since looking at $\log n$ keys is best possible in the worst case. 
Unfortunately, however, other users have accessed files before you and have only returned the files to approximately the right place.
As a result, the register is unsorted, but at least each file is within some small number $k$ of positions of where it should be. 
How should you proceed? 
If you knew $k$ and $n$, at what ratio of $k$ vs.\ $n$ should you resort to a linear search of the register? 
If you do not know $k$, can you still do reasonably well? 
What if the register was recently moved, by packing the files into boxes, but in the process the order of the boxes got mixed up, and now there are large blocks of files that are far away from their correct locations? 
What if you misread some of the keys?
Situations like these are close to searching in a sorted register and there are plenty of parameters that measure closeness to a sorted array, e.g., maximum displacement or minimum block moves to sort, respectively persistent or temporary read errors. We give (almost) optimal algorithms for a large variety of these measures, and thereby establish for each of them exact regimes in which we can outperform a linear search of all elements, or even be almost as good as binary search.

More formally, we study the fundamental topic of comparison-based search, which is central to many algorithms and data structures~\cite{Knuth73BookV3,Mehlhorn84Book,Sedgewick1998Book}. In its most basic form, the search problem can be phrased in terms of locating an element~$e$ within a given array~$A$.
In order to search~$A$ efficiently, we need structure in the ordering of its elements:
In general, we cannot hope to avoid querying all entries to find~$e$.
The most prominent example of an efficient search algorithm that exploits special structure is \emph{binary search} for sorted arrays.
Binary search is best-possible for this case.
It needs only logarithmically many queries and is thus very well suited for searching extremely large collections of data. 
However, it heavily relies on perfect order and reliable access to the data.
For large and dynamically changing collections of data, both requirements may be difficult to ensure, but it may be reasonable to assume the number of imperfections to be bounded.
Accordingly, we ask: 
\emph{What is the best-possible search algorithm if the data may be disordered or we cannot access it reliably?}
\emph{In what regime of the considered measure is it better than linear search?}

We provide (almost) tight bounds on the query complexity of searching an array~$A$ with~$n$ entries for an element~$e$ in a variety of settings.
Each setting is characterized by bounding a different parameter $k$ that quantifies the imperfections regarding either our access to array elements or regarding the overall disorder of the data.
Note that one can always resort to linear search, which rules out lower bounds stronger than~$n$ comparisons.%
\footnote[\textdagger]{Accordingly, all (lower) bounds of the form~$f(n,k)$ throughout the paper are to be understood as~$\min\{f(n,k), n\}$. A naive bound of~$n$ can easily be obtained by scanning the whole array.\label{foot:trivial_bound}}
Table~\ref{tab:results} gives an overview of the parameters we analyze and our respective results.
Qualitatively, our results can be grouped into three groups of settings leading to different query complexities, and we briefly highlight each group in the following.

The first group contains the parameters~$\ksum$, $\kmax$, and~$\kinv$, which quantify the summed/maximum distance of each element from its position in sorted order and the number of element pairs in the wrong relative order, respectively (detailed definitions can be found the the corresponding sections).
For all of these parameters we are able to show that~$\log n + ck$ queries are necessary and sufficient, for constant~$c$.
Intuitively, this is the best complexity we can hope for: We cannot do better than~$\log(n)$ queries, and the impact of~$k$ on the query complexity is linear and can be isolated.

The second group of results is with respect to the parameters~$\klies$, $\kfaults$, as well as multiple parameters for edit distances that measure the number of element operations needed to sort~$A$.
The parameter~$\klies$ limits the number of queries that yield the wrong result, and $\kfaults$ limits the number of array positions that yield wrong query outcomes.
For bounded values of~\klies and~\kfaults we show that $e$ cannot be found with $\log n + ck$ queries using any binary-search-like algorithm.%
\footnote{We interpret the array as a binary tree (rooted at entry $n/2$, with the two children $n/4$, $3n/4$, etc.), and call an algorithm ``binary-search-like'' if it never queries a node (other than the root) before querying its parent.}
On the other hand, we provide an algorithm that needs~$(1+1/c)\log n + ck$ queries, for any~$c \geq 1$.
For bounded edit distances, it is easy to see that we need~$n$ queries if~$e$ need not be at its correct position relative to sorted order, since~$e$ can be moved anywhere with just 2 edits, forcing us to scan the whole array.
If we assume~$e$ to be at its correct location, we can carry over the results for \klies and \kfaults to obtain the same bounds for the edit-distance related parameters \krep, \kseq, \kmov, and \kswap. 

Lastly, we consider the parameter~$\kainv$ that counts the number of adjacent elements that are in the wrong relative order, as well as several parameters measuring the number of block operations needed to sort~$A$.
Intuitively, these settings are much more difficult for a search algorithm, as it takes relatively small parameter values to introduce considerable disorder.
For the case that~$e$ is guaranteed to be at the correct position, we show that~$\sqrt{cnk}+o(nk)$ queries are necessary and sufficient to locate~$e$.

The algorithms for~$\kainv$ and related parameters assume that the parameter value is known to the algorithm a priori.
In contrast, all our other algorithms are oblivious to the parameter, in the sense that they do not require knowledge of the parameter value as long as the target element~$e$ is guaranteed to be present in the array.
Note that if~$e$ need not be present and we have no bound on the disorder, we generally need to inspect every entry of the array in case we cannot find~$e$.
For the parameter \klies, we do not even know how long we need to continue querying the same elements until we may conclude that~$e$ is not part of the array.
Any of our oblivious algorithms can trade the guarantee that $e \in A$ against knowledge of the parameter value $k$: Compute from $k$ the maximum number $m$ of queries that it would take without knowing $k$ when $e\in A$. If the algorithm does not stop within $m$ queries then it is safe to answer that $e$ is not in $A$.

Overall, our results point out several parameters for which a fairly large regime of $k$ (as a function of $n$) allows search algorithms that are provably better than linear search. For example, while moving only a single element by a lot can lead to bounds of $\Omega(n)$ on the values of several parameters, and hence trivial guarantees, moving many elements by at most~$k$ places gives $\kmax=k$ and yields better bounds than linear search (roughly) for $k<\frac{n}{3}$, and as good as binary search when $k=\Oh(\log n)$. Moving only few elements by an arbitrary number of spaces, in turn, still leads to good bounds via parameters such as~\kmov or~\kswap, as long as the target is in the correct place. Parameters such as \kainv grow even more slowly, for certain types of disorder, but, on the other hand, only a small regime allows for better than trivial guarantees. 
While, for each individual parameter we study, there are ``easily searchable'' instances where the parameter becomes large and makes the corresponding bound trivial, our results often allow for good bounds by resorting to a different parameter.

\newcolumntype{P}[1]{>{\centering\arraybackslash}p{#1}}
\newcommand{\myline}[1]{\raisebox{1.5mm}{\underline{\hspace{#1}}}}
\newcommand{\thmref}[1]{\scriptsize{[Th.~\ref{thm:#1}]}}
\newcommand{\corref}[1]{\scriptsize{[Co.~\ref{cor:#1}]}}
\newcommand{\thmrefs}[1]{\scriptsize{[Th.~\ref{thm:LB_#1},\ref{thm:UB_#1}]}}

\begin{table}[t]
	
\caption{Overview of our results, with main results in boldface.\footnoteref{foot:trivial_bound} %
($^o$: even if oblivious to parameter value; $^c$: for all $c\geq1$; $^t$: for tree-algorithms; $^e$: for $\pos(e)=\rank(e)$)}
\label{tab:results}
	
\begin{minipage}{\textwidth}
\begin{center}
\begin{tabular*}{\linewidth}{
@{\extracolsep{1ex}}l 
@{\extracolsep{\fill}}l 
@{\extracolsep{1ex}}   P{3.6cm}
@{\extracolsep{1ex}}   P{5.1cm}
}
	
\toprule
& & \multicolumn{2}{c}{\myline{3cm} bounds \myline{3cm}}\\
\multicolumn{2}{@{}l}{parameter description} & lower & upper\\

\midrule

\multicolumn{3}{l}{\textbf{Section~\ref{sec:imprecise_queries}} -- number of imprecise queries } \\
$\klies$ & wrong outcomes & $\boldsymbol{\log n + ck}$ \thmref{LB_klies}$^{ct}$ &  $\boldsymbol{(1\!+\!\frac{1}{c})\log n + (2c\!+\!2)k}$\,\thmref{UB_klies}$^{oc}$\\
$\kfaults$ & indices with wrong outcomes & $\log n + ck$ \corref{LB_kfaults}$^{ct}$ & $(1\!+\!\frac{1}{c})\log n + (2c\!+\!2)k$ \thmref{UB_kfaults}$^{oc}$ \\

\midrule

\multicolumn{3}{l}{\textbf{Section~\ref{sub:displacement}} -- displacement of elements} \\
$\ksum$ & total displacement & \multicolumn{2}{c}{$\boldsymbol{\log n/k + 2k + \Oh(1)}$ \thmrefs{ksum}$^o$} \\
$\kmax$ & maximum displacement & \multicolumn{2}{c}{$\boldsymbol{\log n/k + 3k + \Oh(1)}$ \thmrefs{kmax}$^o$} \\

\midrule

\multicolumn{3}{l}{\textbf{Section~\ref{sub:inversions}} -- number of inversions} \\
$\kinv$ & all inversions & $\log n/k\!+\!2k\!+\!\Oh(1)$ \corref{LB_kinv} & $\log n/k + 4k + \Oh(1)$ \corref{LB_kinv}$^o$ \\
$\kainv$ & adjacent inversions & \multicolumn{2}{c}{$\boldsymbol{\sqrt{8nk} + o(\sqrt{nk})}$ \thmrefs{kainv}$^e$}  \\

\midrule

\multicolumn{3}{l}{\textbf{Section~\ref{sub:edit_distances}} -- element operations needed to sort the array} \\
$\krep$ & element replacements & $\log n + ck$ \corref{LB_krep}$^{cte}$ & $(1\!+\!\frac{1}{c})\log n + (4c\!+\!4)k$ \thmref{UB_krep}$^{oe}$ \\

$\kseq$ & $n - |\textrm{max ordered subseq.}|$ & $\log n + ck$ \corref{LB_kseq}$^{cte}$ & $(1\!+\!\frac{1}{c})\log n + (4c\!+\!4)k$ \thmref{UB_kseq}$^{oe}$ \\

$\kmov$ & element moves & $\log n + ck$ \corref{LB_kmov}$^{cte}$ & $(1\!+\!\frac{1}{c})\log n + (4c\!+\!4)k$ \thmref{UB_kmov}$^{oe}$ \\

$\kswap$ & element swaps & $\log n + ck$ \corref{LB_kswap}$^{cte}$ & $(1\!+\!\frac{1}{c})\log n + (8c\!+\!8)k$ \thmref{UB_kswap}$^{oe}$ \\

$\kaswap$ & adj. element swaps & $\log n/k\!+\!2k\!+\!\Oh(1)$ \corref{LB_kaswap} & $\log n/k + 4k + \Oh(1)$ \corref{UB_kaswap}$^o$ \\

\midrule

\multicolumn{3}{l}{\textbf{Section~\ref{sub:block_distances}} -- block operations needed to sort the array} \\
$\kbswap$ & block swaps & \multicolumn{2}{c}{$4\sqrt{nk} + o(\sqrt{nk})$ \corref{UB_kbswap} \thmref{LB_kbswap}$^e$} \\
$\krbswap\,$ & equal size block swaps & $2\sqrt{2nk} + o(\sqrt{nk})$ \corref{UB_krbswap}$^e$ & $4\sqrt{nk} + o(\sqrt{nk})$ \thmref{LB_krbswap}$^e$ \\
$\kbmov$ & block moves & \multicolumn{2}{c}{$2\sqrt{2nk} + o(\sqrt{nk})$ \thmref{UB_kbmov}$^e$ \corref{LB_kbmov}} \\

\bottomrule
	
\end{tabular*}
\end{center}
\end{minipage}
\end{table}


\subsection{Related Work}

Our work falls into the area of \emph{adaptive analysis of algorithms}, which aims at a fine-grained analysis of polynomial-time algorithms with respect to structural parameters of the input. 
An objective of this field is to find algorithms whose running-time dependence on input size and the structural parameters interpolates smoothly between known (good) bounds for special cases and the worst-case bound for general inputs. 
The topic of adaptive sorting, i.e., sorting arrays that are presorted in some sense, has attracted a lot of attention, see, e.g.,~\cite{BarbayN13, Estivill-CastroW92, Mehlhorn79b, PeterssonM95}. 

We now discuss results that are specific to searching in arrays. 
Several authors addressed the question of how much preprocessing, i.e., sorting, helps for searching, if we take into account the total time investment~\cite{BorodinGLY81, Mairson85, PetrankR04}. 
Fredman~\cite{Fredman03} gave lower bounds on searching regarding both queries and memory accesses. 
A classic work of Yao~\cite{Yao81a} established that the best way of storing~$n$ elements in a table such as to minimize number of queries for accessing an element is by keeping the elements sorted, which requires $\log n$ queries, provided that the key space is large enough.
Regarding searching in (partially) unordered arrays, there is a nice result of Biedl et al.~\cite{BiedlCDFGKM04} about insertion sort based on repeated binary searches. 

Under appropriate assumptions, namely that array is sorted and its elements are drawn from a known distribution (e.g., searching for a name in a telephone book), one can do much better than binary search, since the distribution allows a good prediction of where the target should be located. 
In this case $\Oh(\log\log n)$ queries suffice on average (cf.\ \cite{Sedgewick1998Book}); to avoid having to query the entire array, previous work suggests combinations of algorithms that perform no worse than binary search in the worst case~\cite{BurtonL80,BonaseraFFPP15}. 
Another interesting branch of study is related to search in arrays of more complicated objects such as (long) strings~\cite{AnderssonHHP00, FranceschiniG08} or abstract objects with nonuniform comparison cost~\cite{GuptaK01, AngelovKM08}.

Many papers have studied searching in the presence of different types of errors, e.g.,~\cite{BorgstromK93, FinocchiGI09, FinocchiI08, Muthukrishnan94}, see~\cite{Cicalese13, Pelc02} for surveys.
A popular error model for searching allows for a linear number of lies~\cite{AslamDhagat91, BorgstromK93, DhagatGW92, FeigeRPU94, Pelc89b}, for which Borgstrom and Kosaraju~\cite{BorgstromK93} gave an~$\Oh(\log n)$ search algorithm.
In constrast, we bound the number of lies separately via the parameter~$\klies$.
Rivest et al.~\cite{RivestMKWS80} gave an upper bound of~$\log n + k \log \log n + \Oh(k \log k)$ queries for this parameter.
Their algorithm is based on a continuous strategy for the (equivalent) problem of finding an unknown value in~$[1,n]$, upto a given precision, using few yes-no questions.
Our algorithm (Theorem~\ref{thm:UB_klies}) uses asymptotically fewer queries if~$\klies = \omega(\log n / \log \log n)$.\footnote{A technical report of Long~\cite{Long92} claims that the actual tight bound of the algorithm of Rivest et al.~\cite{RivestMKWS80} is $\Oh(\log n + k)$, which is consistent with our results.}

The works of Finocchi and Italiano~\cite{FinocchiI08} and Finocchi et al.~\cite{FinocchiGI09} consider a parameter very similar to \kfaults, with the additional assumption that faults may affect also the working memory of the algorithm, except for~$\Oh(1)$ ``safe'' memory words.
Finocchi and Italiano~\cite{FinocchiI08} give a deterministic searching algorithm that needs~$\Oh(\log n + k^2)$ queries.
Brodal et al.~\cite{BrodalFagerbergFinocchiGrandoniItalianoJorgensenMoruzMolhave/07} improve this bound to~$\Oh(\log n + k)$ and Finocchi et al.~\cite{FinocchiGI09} provide a lower bound of~$\Omega(\log n + k)$ even for randomized algorithms. 
Our results are incomparable as our result for parameter \kfaults uses only $(1+\frac{1}{c})\log n + (2c+2)k$ queries, getting arbitrarily close to $\log n + O(k)$ (cf.~Theorem~\ref{thm:UB_kfaults}), but does not consider faults in the working memory; the high level approach of balancing progress in the search with security queries is the same as in~\cite{BrodalFagerbergFinocchiGrandoniItalianoJorgensenMoruzMolhave/07}, but more careful counting is needed to get small constants.
For parameter \klies we give a simpler algorithm with $2\log n + 4k$ queries and using only $\Oh(1)$ words of working memory, but it is not clear whether the result can be transferred to \kfaults without increasing the memory usage.

Finally, we comment on the measures of disorder we adopt in this paper. 
We study various well-known measures that are mostly folklore. 
Detailed overviews of measures and their relations were given by Petersson and Moffat~\cite{PeterssonM95} and Estivill-Castro and Wood~\cite{Estivill-CastroW92}. 
For the sake of completeness and to get all involved coefficients, we have provided in Appendix~\ref{section:relations:disorder} proofs of all pairwise relations between our parameters, as depicted in Figure~\ref{fig:relations}.


\section{Preliminaries}
In this paper we consider the following problem: Given an array~$A$ of length~$n$ and an element~$e$, find the position of~$e$ in~$A$ or report that~$e \notin A$ with as few \emph{queries} as possible. We use~$A[i]$, $i \in {1, \dots, n}$ to denote the~$i$-th entry of~$A$. We allow access to the entries of~$A$ only via queries to its indices, regarding the relation of the corresponding element to~$e$. We write~$\query(i)$ for the operation of querying~$A$ at index~$i$, and let~$\query(i) =\,$`<' (respectively, `>' or `=') denote the outcome indicating that~$A[i] < e$ (respectively~$A[i]>e$ or $A[i]=e$). Note that in faulty settings the query outcome need not be accurate. 

To keep notation simple, we generally assume the entries of~$A$ to be unique unless explicitly stated otherwise. We emphasize that none of our results relies on this assumption. We can then define~$\pos(a)$ to denote the index of~$a$ in~$A$, by setting $\pos(a) = i$ if and only if~$A[i] = a$. Further, let~$\rank(a) = |\{i:A[i] < a\}| + 1$ be the ``correct'' position of~$a$ with respect to a sorted copy of~$A$, irrespective of whether or not~$a \in A$. We often use an element~$a \in A$ and its index~$\pos(a)$ interchangeably, especially for the target element~$e$. 
Note that, as discussed in the introduction, for oblivious algorithms we generally assume $e\in A$.


\section{Searching with imprecise queries}\label{sec:imprecise_queries}

In this section, we consider the problem of finding the index~$\pos(e)$ of an element~$e$ in a sorted array~$A$ of length~$n=2^d$, $d \in \N$ in a setting where queries may yield erroneous results. 
We say that `<' is a lie (the truth) for index $i$ if $A[i] \geq e$ ($A[i]<e)$, and analogously for `>' and `='.
To quantify the number of lies, we introduce two parameters~$\klies$ and~$\kfaults$.
The first parameter~$\klies$ simply bounds the number of queries with erroneous results, which we interpret as the number of lies allowed to an adversary.
The second parameter~$\kfaults$ bounds the number of indices~$i$ for which~$\query(i)$ (consistently) returns the wrong result, allowing the conclusion that $e \notin A$ in case $\query(e)$ yields the wrong result.
Equivalently, for an unsorted array~$A$, we can require all queries to be truthful and define~$\kfaults(e)$ to be the number of inversions involving~$e$, i.e., $\kfaults(e) = |{i: (i<\pos(e) \wedge A[i] > e) \vee (i>\pos(e) \wedge A[i] < e)}|$.
Observe that both definitions of \kfaults are equivalent.
For clarity, we write~$\kfaults$ when considering the adversarial interpretation, and~$\kfaults(e)$ when considering it as a measure of disorder of an unsorted array.
For both \klies and \kfaults, we only allow queries to $e$ to yield `='.

The algorithms of this section operate on the binary search tree rooted at index $r = n/2$ that contains a path for each possible sequence of queries in a binary search of the array, and identify nodes of the tree with their corresponding indices. 
We write $\nxt_{>}(i)$ and $\nxt_{<}(i)$ to denote the two successors of node $i$, e.g., $\nxt_{>}(r) = n/4$ and $\nxt_{<}(r) = 3n/4$. 
Similarly, we write~$\prev(i)$ to denote the predecessor of $i$ in the binary search tree, and $\prev_{q}(i) = v$ for the last vertex~$v$ on the unique $r$-$i$-path such that $\nxt_{q}(v)$ also lies on the $r$-$i$-path ($\prev_{q}(i)=\emptyset$ if no such node exists). 
Intuitively, $\prev_{q}(i)$ is the last vertex corresponding to an array entry larger (if $q = $``>'') or smaller (if $q = $``<'') than~$A[i]$. 
For convenience, $\query(\emptyset)=\emptyset$, $\prev_{</>}(r)=\emptyset$, and $\nxt_{</>}(i)=i$ if $i$ is a leaf of the tree. 
We further denote by $d(i,j)$ the length of the path from node~$i$ to node~$j$ in the search tree.
We say that an algorithm \emph{operates on the binary search tree} if no index is queried before its predecessor in the tree.

We start by considering the parameter~$\klies$.
If we knew the value of this parameter, we could try a regular binary search, replace every query with~$2\klies + 1$ queries to the same element and use the majority outcome in each step.
However, this would give~$(2\klies+1)\log n$ queries, where ideally we should not use more than~$\log n + f(k)$ queries.
We first give an algorithm that achieves the separation between~$n$ and $\klies$ while being oblivious to the value of~$\klies$.
Importantly, the algorithm only needs~$\Oh(1)$ memory words, which also makes it applicable to settings where ``safe'' memory, that cannot be corrupted during the course of the algorithm, is limited.
This algorithm still needs~$2\log n + f(k)$ queries, but we will show later how to build on the same ideas to (almost) eliminate the factor of~$2$.

\newcommand\mycommfont[1]{\footnotesize\textcolor{gray}{#1}}
\SetCommentSty{mycommfont}

\begin{algorithm}[ht]
\caption{Algorithm with $2 \log n + 4 \klies$ queries}\label{alg:lies}
\DontPrintSemicolon
\emph{the algorithm stops once a query yields} `$=$'\;

 $i \leftarrow n/2$ \tcp*{start at the root}
 
 \While(\tcp*[f]{by definition, `$=$' cannot be a lie}){$(q \leftarrow \query(i)) \neq $`$=$'} 
 {
   $i' \leftarrow \prev_{\neg q}(i)$ \tcp*{$\emptyset$ if all queries on the path from the root yielded $q$} 
   \While(\tcp*[f]{while $\query(i')$ contradicts its previous outcome\ldots}){$i \neq i' \wedge \query(i') = q$}
   {
     $i \leftarrow \prev(i)$ \tcp*{\ldots backtrack towards $i'$}
   }
   \If(\tcp*[f]{if we did not backtrack all the way to $i'$\ldots}){$i \neq i'$}
   { $i \leftarrow \nxt_q(i)$ \tcp*{\ldots proceed according to $q$} }
 }
\end{algorithm}

Intuitively, Algorithm~~\ref{alg:lies} searches the binary search tree defined above, simply proceeding according to the query outcome at each node.
In addition, the algorithm invests queries to double check past decisions.
We distinguish left and right turns, depending on whether the algorithm proceeds with the left or the right child.
In particular, before proceeding, the algorithm queries the last vertex on the path from the root where it decided for a turn in the opposite direction.
While an inconsistency to previous queries is detected, i.e., a query to a vertex where it turned right (or left) gives `>' (or `<'), the algorithm backtracks one step.
In this manner, the algorithm guarantees that it never proceeds along a wrong path without the adversary investing additional lies.
Note that if the algorithm only ever turned right (or left), i.e., there was no previous turn in the opposing direction, it does not double check any past decisions until the query outcome changes.
This is alright since either the algorithm is on the right path or the adversary needs to invest a lie in each step.

\begin{theorem}
We can find $e$ obliviously using $2\log n+4\klies$ queries and~$\Oh(1)$ memory. \label{thm:klies_upper}
\end{theorem}

\begin{proof}
We claim that Algorithm~\ref{alg:lies} achieves the bound of the theorem. 
Note that~$\prev_{\neg q}(i)$ only depends on~$i$ and not on the outcome of previous queries, therefore, we can determine it with~$\Oh(1)$ memory words.
We will show that in each iteration of the outer loop of the algorithm, the potential function $\Phi=2d(i,e)+4k$ 
decreases by at least one for each query, where $k$ is the number of remaining lies the adversary may make. 
This proves the claim, since $\Phi\geq0$ and initially $\Phi\leq2\log n+4\klies$. 
We analyze a single iteration of the outer loop. 

Observe that if $z$ is the number of iterations of the inner loop,
then the total number of queries is $z+2$ if the inner loop terminates
because $\query(i')=\neg q$, and $z+1$ if it terminates
because $i=i'$. If an iteration of the inner loop is caused by $\query(i')$
being a lie, then in this iteration $\Delta\Phi\leq2-4=-2$, and otherwise,
$d(i,e)$ is decreased by one and likewise $\Delta\Phi=-2+0=-2$.
Overall, the change in potential during the inner loop is always $\Delta\Phi=-2z$.
If the inner loop terminates because $i=i'$, then $z\geq1$ and the
total change in potential is $\Delta\Phi\leq-2z\leq-z-1$, enough
to cover all $z+1$ queries. 

Now consider the case that the inner loop terminates because $\query(i')=\neg q$.
If $\neg q$ is a lie for $i'$ or $q$ is a lie for $i$, the adversary
invested an additional lie, and even if the last update to $i$ increases
$d(i,e)$, the total change in potential is bounded by $\Delta\Phi\leq-2z-4+2\leq-2z-2$,
enough to cover all $z+2$ queries. On the other hand, if $\neg q$
is the truth for $i'$ and $q$ is the truth for $i$, then $e\in\{i',\dots,i\}$ and $i$ must lie on the unique $r$-$e$-path in the search tree
(and $i\neq e$). The final update to $i$ thus decreases $d(i,e)$
by 1 and the total change in potential is $\Delta\Phi=-2z-2$, again
enough to cover all $z+2$ queries.
\end{proof}

\begin{algorithm}[ht]
\caption{Algorithm with $(1+\frac{1}{c})\log n + (2c+2)\klies$ queries}\label{alg:lies_eps}
\DontPrintSemicolon
\emph{the algorithm stops once a query yields} `$=$'\;
 $i \leftarrow n/2$ \tcp*{start at the root}
 \While(\tcp*[f]{by definition, `$=$' cannot be a lie}){$(q \leftarrow \query(i)) \neq $`$=$'}
 {
   $i' \leftarrow \prev_{\neg q}(i)$ \tcp*{$\emptyset$ if all queries on the path from the root yielded $q$}  
   \While(\tcp*[f]{while we do not have sufficient support to proceed\dots}){$0 < c \Delta_{i'} < d(i,i')+1$}
   {
     $\query(i')$ \tcp*{\dots query $i'$ for support}
   }
   \uIf(\tcp*[f]{if we ran out of support at $i'$ altogether\dots}){$\Delta_{i'} = 0$}
   { $i \leftarrow i'$ \tcp*{\dots backtrack to $i'$} }
   \Else(\tcp*[f]{if we have sufficient support at $i'$\dots})
   { $i \leftarrow \nxt_q(i)$ \tcp*{\dots proceed according to $q$} }
 }
\end{algorithm}

We now adapt Algorithm~\ref{alg:lies} to minimize the impact of potential lies on the dependency on~$\log n$ in the running time. 
Intuitively, instead of backing up each query~$q \leftarrow \query(i)$ by a query to~$\prev_{\neg q}(i)$, we back only one in~$c$ queries (cf.~Algorithm~\ref{alg:lies_eps}).
During the course of the algorithm and its analysis, we let $n_{q,j}$ denote the number of queries (so far) to node~$j$ that resulted in $q\in\{<,>\}$ and $\Delta_{j}:=\left|n_{<,j}-n_{>,j}\right|$.

\begin{theorem}
For every $c \geq 1$, we can find $e$ obliviously using $(1+\frac{1}{c})\log n+(2c+2)\klies$ queries.\label{thm:UB_klies}
\end{theorem}

\begin{proof}
We claim that Algorithm~\ref{alg:lies_eps} achieves the bound of the theorem. 
In this algorithm, we intuitively back up every $c$-th query (for integral~$c$).
To capture this in our potential function, we need a term that stores potential for the next~$c$ queries.
We will introduce two such terms $L,T$, representing the case the algorithm's current belief of the relation between~$i'$ and~$e$ is a lie or the truth, respectively.
We need to distinguish these cases, since they lead to different behavior regarding~$d(i,e)$ and the number of remaining lies.

We need the following additional notation. 
For some current value of $i$ during the execution of the algorithm, we define the type of a node $j$ of the search tree on the $r$-$i$-path to be $t_{j}\in\{<,>\}$ if $\nxt_{t_{j}}(j)$ also lies on this path. 
Further, we let $\scc_{q}(j)=j'$ if $j'$ is the first node on the $j$-$i$-path with $t_{j'}=q$ or $\scc_{q}(j)=i$ if no such node exists. 
We set $\scc_{q}(i)=\emptyset$.
To avoid special treatment of leaves, we replace each leaf of the search tree by an infinite binary tree of nodes corresponding to the original leaf, in both algorithm and analysis. 
If $e$ was a leaf, then, for each new node $j$ corresponding to $e$, we set $d(j,e)=d(j,r_{e})$ where $r_{e}$ is the root of the subtree corresponding to~$e$.

Intuitively, the potential of the algorithm needs to depend on~$c\Delta_{i'} - d(i,i')$, since this difference captures the number of steps it can still make before it needs to use a backup query.
To keep track of this difference across iterations of the algorithm, we introduce the notion of a \emph{zig-zag pair}, which we will define formally below.
In particular, $(i,i')$ always forms a zig-zag pair.
Let $j=\nxt_q(i)$ in some iteration after which~$\Delta_{i'} \neq 0$, i.e., $i$ gets updated to~$j$.
If~$j$ has the same type as~$i$ in the next iteration, $i'$ stays the same and we can simply replace the zig-zag pair~$(i,i')$ with~$(j,i')$.
On the other hand, if~$j$ has a different type in the next iteration, we need to introduce a new pair~$(j,i)$.
Since we may backtrack later and continue differently at~$i$, we also need to keep the pair~$(i,i')$.
Conceptually, we need to keep track of all maximal $l$-$l'$-subpaths of the current~$r$-$i$-path with the property that $l'$ has the opposite type than all other nodes on the subpath.
If the algorithm backtracks to node~$l$ at some point, then, in the next iteration, $i=l$ and $i'=l'$, and the difference~$c\Delta_{l'} - d(l,l')$ captures how much potential remains to continue querying without using a back up query to~$i'=l'$.

Formally, we define the set of all zig-zag pairs as 
\[
Z:=\left\{ (j,j') \mid \exists q\in\{<,>\}.j'=\prev_{q}(j)\wedge j=\scc_{q}(j')\right\}.
\]
Note that~$(i,i') \in Z$ throughout the algorithm, and that every node appears at most once as the second element of a zig-zag pair, exactly if it has a different type than its successor on the unique~$r$-$i$-path.
For convenience, we set $\Delta_{\emptyset}$ to be equal to the number
of all ``queries'' to $\emptyset$ in the algorithm. 
We define
\[
L=\sum_{(j,j')\in Z}[c\Delta_{j'}-d(j,j')]\cdot\Lambda_{t_{j'},j'},
\]
where $\Lambda_{q,j}=1$ if $q$ is a lie for $j$ and
$\Lambda_{q,j}=0$ otherwise. 
Similarly, we define 
\[
T=\sum_{(j,j')\in Z}[d(j,j')-c(\Delta_{j'}-1)]\cdot(1-\Lambda_{t_{j'},j'}).
\]

With this notation in place, we introduce the extended potential function
\[
\Phi=(1+\frac{1}{c})d(i,e)+(2+\frac{1}{c})L+\frac{1}{c}T+(2c+2)k,
\]
where $k$ is the number of lies remaining to the adversary.

We claim that~$L,T \geq 0$ and thus $\Phi \geq 0$ holds after each iteration.
More precisely, we show that the contribution of each zig-zag pair in~$Z$ to either~$L$ or~$T$ is non-negative.
To see this, first observe that in each iteration~$Z$ changes exactly by either removing the zig-zag pair~$(i,i')$ (unless~$i=r$), by replacing it with the pair~$(\nxt_q(i),i')$, or by adding a new pair~$(\nxt_q(i),i)$.
Inductively, it thus suffices to show that the contribution of~$(\nxt_q(i),i')$ or~$(\nxt_q(i),i)$ in the latter cases ($\Delta_{i'} \neq 0$) is positive.
First, observe that~$\Delta_{i} = 1$ after the iteration, hence, if $(\nxt_q(i),i)\in Z$, its contribution to~$L$ or~$T$ must be positive.

Now consider the case that $(\nxt_q(i),i')\in Z$ after the iteration.
By definition of the algorithm, the inner loop ensures that
\[
c\Delta_{i'} \geq d(i,i') + 1 = d(\nxt_q(i),i'),
\]
hence the contribution of~$(\nxt_q(i), i')$ to~$L$ is non-negative.
Now consider the last iteration of the outer loop in which~$\Delta_{i'}$ changed, and let~$j,j'$ be the corresponding values of $i$ and $i'$ in that iteration.
Either~$d(i,i') = \Delta_{i'} = 1$, or the last change to~$\Delta_{i'}$ was because~$j' = i'$ and~$c\Delta_{i'} < d(j,i') \leq d(i,i') + 1$.
In the latter case, after the update to~$\Delta_{i'}$, we have~$c\Delta_{i'} \leq d(j,i') + 1 + c$ and thus, in both cases,
\[
c\Delta_{i'} \leq d(i,i') + c < d(\nxt_q(i),i') + c,
\]
hence the contribution of~$(\nxt_q(i), i')$ to~$T$ is non-negative.

Initially, $T=L=0$ since $Z=\emptyset$, and $\Phi=(1+\frac{1}{c})\mathrm{log}n+(2c+2)\klies$.
Since $\Phi \geq 0$ throughout, it thus suffices to show that in each iteration of the outer loop $\Phi$ decreases by at least one for each query. 
We consider a single iteration of the outer loop for fixed $i,i'$.

First consider the case where the inner loop is not executed. In this
case, the algorithm makes only a single query, and we need to show
that the potential decreases by at least 1. Observe that $L$ cannot
increase during the update to $i$ and $T$ may increase by at most
1. If $q$ is a lie for $i$, then this update increases $d(i,e)$
by 1 and the adversary invested an additional lie, for a change in
potential of at most $\Delta\Phi\leq(1+\frac{1}{c})+0+\frac{1}{c}-(2c+2)=-2c-1+\frac{2}{c}\leq-1$ (since~$c \geq 1$).
If $\neg q$ is a lie for $i'$, then $d(i,e)$ increases by 1 and
$L$ decreases by 1, for a change in potential of $\Delta\Phi\leq(1+\frac{1}{c})-(2+\frac{1}{c})+0+0=-1$.
If $q$ is the truth for $i$ and $\neg q$ is the truth for $i'$,
then $e\in \{i,\dots,i'\}$ and $i$ must lie on the unique $r$-$e$-path.
The update to $i$ then decreases $d(i,e)$ and changes the potential
by $\Delta\Phi=-(1+\frac{1}{c})+0+\frac{1}{c}-0=-1$.

Now consider the case where the inner loop is executed until~$\Delta_{i'} = 0$, and fix the value of $\Delta_{i'}$ before the inner loop.
We may assume that no query to $i'$ yielded $\neg q$, otherwise
we can balance each such query with a query that yielded $q$, one
the two being a lie, for a change in potential of $\Delta\Phi=-(2c+2)\leq-4$,
which pays for both these queries. With this assumption, we have exactly
$\Delta_{i'}\geq1$ queries in the loop, all of which yielded $q$.
Note that the previous iteration of the outer loop ensured that $c\Delta_{i'} \geq d(i,i')$.
If~$\neg q$ is a lie for $i'$, the eventual update to $i$ decreases
$d(i,e)$ by $d(i,i')$ and decreases $L$ by $c\Delta_{i'}-d(i,i')$ (since $(i,i')$ is eliminated from $Z$). 
The overall change in potential
then is
\begin{eqnarray*}
\Delta\Phi & \leq & -(1+\frac{1}{c})\cdot d(i,i')-(2+\frac{1}{c})[c\Delta_{i'}-d(i,i')]+0+0\\
 & = & d(i,i')-c\Delta_{i'}-(c+1)\Delta_{i'}\\
 & \leq & -(c+1)\Delta_{i'}\\
 & \overset{c \geq 1}{\leq} & -1-\Delta_{i'},
\end{eqnarray*}
which is enough to cover all $1+\Delta_{i'}$ queries. On the other
hand, if $\neg q$ is the truth for $i'$, the eventual update may
increase $d(i,e)$ by at most $d(i,i')$ and it eliminates the contribution 
of $(i,i')$ to $T$ (since $i=i'$ and, hence, $(i,i')\notin Z$). The
adversary invested $\Delta_{i'}$ additional lies, and the change
in potential is
\begin{eqnarray*}
\Delta\Phi & \leq & +(1+\frac{1}{c})d(i,i')+0+0-(2c+2)\Delta_{i'}\\
 & \overset{c \geq 1}{\leq} & 2d(i,i')-4\Delta_{i'}\\
 & \leq & -2\Delta_{i'}\\
 & \leq & -1-\Delta_{i'},
\end{eqnarray*}
which is again enough to cover all $1+\Delta_{i'}$ queries.

Finally, consider the case where the inner loop is executed until
$c\Delta_{i'}\geq d(i,i')+1$. As before, $c\Delta_{i'} \geq d(i,i') \geq 1$,
and, again, we may assume that no query to $i'$ yielded $q$. Hence,
as $c\geq1$, we made a single query to $i'$ that yielded $\neg q$,
i.e., 2 queries overall. We need to show that $\Delta\Phi\leq-2$.
Assume first that $\neg q$ is the truth for $i'$ and we thus decreased
$T$ by $(c-1)$. If $q$ is the truth for $i$, we have $\Delta\Phi\leq-(1+\frac{1}{c})+0-\frac{1}{c}(c-1)-0=-2$.
If $q$ is a lie for $i$, we have $\Delta\Phi\leq(1+\frac{1}{c})+0-\frac{1}{c}(c-1)-(2c+2)=-2-2c+\frac{2}{c}\leq-2$.
Now assume that $\neg q$ is a lie for $i'$ and we thus increased
$L$ by $(c-1)$. Since the adversary invested an additional lie,
we have $\Delta\Phi\leq(1+\frac{1}{c})+(2+\frac{1}{c})(c-1)+0-(2c+2)=-2$.
\end{proof}

To provide a strong lower bound, we restrict ourselves to algorithms that operate on the binary search tree.
Such algorithms interpret the array as a binary tree (rooted at entry n/2, with the two children n/4, 3n/4, etc.), and never query a node (other than the root) before querying its parent.

\begin{theorem}\label{thm:LB_klies}
For every $c\in\mathbb{N}$, no algorithm operating on the search tree can find $e$ with less than $\log n+c\klies$ queries in general.\footnoteref{foot:trivial_bound}
\end{theorem}

\begin{proof}
We consider the behavior of the algorithm on the search tree for large
values of $n$ that are powers of two. We split the queries of the
algorithm into phases, where phase $p$ starts as soon as a node of
depth $(c+1)\cdot p$ is queried for the first time, starting with
phase~$0$. We take the perspective of an adversary and specify the outcome
to each query, ensuring that at most one lie is invested in each phase,
and at most $\klies$ lies overall. Note that we do not
have to decide immediately whether a query outcome is truthful and
neither do we have to fix the position of $e$ a priori.

Consider a fixed phase $p$. The first query of the phase to node
$i$ of depth $(c+1)\cdot p$ yields `$<$', all subsequent queries
to positions smaller (larger) than $i$ yield `$>$' (`$<$'). If the
algorithm queries more than once a node of depth $(c+1)\cdot p$ in
the first $c+1$ queries of the phase or any node in the left subtree
of $i$, then the phase needs at least $c+2$ queries and we do not
lie, i.e., $e$ is in the subtree rooted at the leftmost node of depth
$(c+1)(p+1)$ in the right subtree of $i$. Otherwise, we lied for
the query to node $i$ and $e$ is in the subtree rooted at the rightmost
node $i'$ of depth $(c+1)(p+1)$ in the left subtree of $i$. Since
no node in the left subtree of $i$ has been queried yet, the algorithm
needs an additional $c$ queries to reach~$i'$, for a total of $2c+1$
queries in the phase. Once all lies have been used up, we continue
answering queries as before, and each phase trivially needs at least
$c+1$ queries.

Observe that querying every node on the path from node $n/2$ to $e$
requires exactly $\log n$ queries, or $c+1$ queries per phase (except
maybe a last, partial phase). Now if we use up all $\klies$
lies, then there are $\klies$ phases that need $c$ additional
queries each, for a total of $\log n+c\klies$, as claimed.
Otherwise, let $P(n)>\left\lfloor \frac{\log n}{c+1}\right\rfloor -\klies$
be the number of phases in which we did not lie. Each such phase needed
$c+2$ queries instead of $c+1$. Overall, we have more than $\log n+P(n)$
queries. Since $P(n)$ is unbounded with growing $n$ while $\klies$
and $c$ are constant, we have $\log n+P(n)\geq\log n+c\klies$
for $n$ large enough, as claimed.
\end{proof}

Note that the construction in the proof of Theorem~\ref{thm:LB_klies} can be applied without change to~$\kfaults$, since the adversary never gives conflicting replies. As a consequence, we immediately obtain a lower bound for~$\kfaults$.

\begin{corollary}
For every $c\in\mathbb{N}$, no algorithm operating on the search tree can find $e$ with less than $\log n+c\kfaults$ queries in general.\label{cor:LB_kfaults}
\end{corollary}

We show how to translate any algorithm with a performance guarantee with respect to~\klies to an algorithm with the same guarantee for~\kfaults.

\begin{theorem} \label{thm:UB_kfaults}
Let $f:\mathbb{N}^{2}\to\mathbb{N}$. If we can find $e$ with $f(n,\klies)$
queries, then we can find $e$ with $f(n,\kfaults)$ queries.
\end{theorem}

\begin{proof}
Assume we have an algorithm that needs $f(n,\klies)$ queries. 
The difficulty when applying this algorithm for $\kfaults$ is that, in the faulty setting, there is no benefit in querying the same elements again. 
However, we can simulate repeated queries to the same element
as follows. 
Say the algorithm needs to query a previously queried element $i$ with the understanding that the adversary has to pay for lying repeatedly. 
Let $i'$ be the first unqueried index to the left or to the right of $i$. If no such index exists, we already queried all elements and found $e$, since $\query(e)$ is guaranteed
to return the correct result. 
We query $i'$ instead of $i$. If $i'=e$, we are done. 
Otherwise, we know that no index in $[i,i']$ contains $e$, and, hence, all these elements are left of $e$ or all of them are right of $e$. 
Therefore the query to $i'$ is equivalent to another query to $i$ when the adversary has to pay for repeated lies. 
Every fault can be treated as a lie by the adversary, and we get the claimed bound.
\end{proof}


\section{Searching disordered arrays}\label{sec:disorder}

In this section, we consider the problem of finding the index~$\pos(e)$ of an element~$e$ in  array~$A$ of length~$n=2^d$, $d \in \N$.
In contrast to Section~\ref{sec:imprecise_queries}, we do not assume~$A$ to be sorted but expect all queries to yield correct results.
We study a variety of parameters that quantify the disorder of~$A$ and provide algorithms and lower bounds with respect to the different parameters.
Figure~\ref{fig:relations} gives the relationship between every pair of parameters.
The proofs of these relationships can be found in Appendix~\ref{sec:relations}.

\begin{figure}
	\centering
	
\begin{tikzpicture}
 \node[draw, rectangle] (kainv) at (0,1.3) {\kainv};
 \node[draw, rectangle] (kbswap) at (2,0) {\kbswap};
 \node[draw, rectangle] (krbswap) at (4,0) {\krbswap};
 \node[draw, rectangle] (kswap) at (6,0) {\kswap};
 \node[draw, rectangle] (kaswap) at (8.4,0) {\kaswap=\kinv};
 \node[draw, rectangle] (ksum) at (10.8,0) {\ksum};
 \node[draw, rectangle] (kbmov) at (3.8,1.3) {\kbmov};
 \node[draw, rectangle] (krep) at (5.3,2.6) {\krep=\kmov=\kseq};
 \node[draw, rectangle] (kmax) at (8.4,2.6) {\kmax};
 
 \tikzstyle{solid}=[->, thick]
 
 \def \off {1.3}
 \path[draw, solid] (kainv.-18) -> node [above] {2} (kbswap);
 \path[draw, solid] (kbswap) -> node {} (krbswap);
 \path[draw, solid] (krbswap) -> node {} (kswap);
 \path[draw, solid] (kswap) -> node {} (kaswap);
 \path[draw, solid] ([yshift=\off]kaswap.east) -> node {} ([yshift=\off]ksum.west);
 \path[draw, solid] (kainv) -> node {} (krep.west);
 \path[draw, solid] (krep) -> node {} (kaswap);
 \path[draw, solid] (kbswap) -> node {} (kbmov);
 \path[draw, solid] (kbmov) -> node {} (krep);
 \path[draw, solid] (krep) -> node [right] {2} (kswap);
 \path[draw, solid] (kbmov.200) -> node [above] {2} (kbswap.70);
 \path[draw, solid] ([yshift=-\off]ksum.west) -> node [above] {2} ([yshift=-\off]kaswap.east);
 \path[draw, solid] (kmax) -> node {} (kaswap);
 
 \tikzstyle{red}=[->, thick, color=red, dashed]
 \path[draw] (kainv.90) edge[bend left = 30, red] node {} (kmax);
 \path[draw, red] (kmax) -> (krep);
 \path[draw, red] (kmax) -> (kswap);
 \path[draw, red] (krep) -> (krbswap);
 \path[draw, red] (kbswap.156) -> (kainv);
 \path[draw, red] (krbswap) -> (krep);
\end{tikzpicture}

  \caption{Overview of relations between measures of disorder. A solid black path from~$k$ to~$k'$ means that $k \leq ck'$, where~$c$ is the product of the edge labels along the path ($c=1$ for unlabeled paths). If there is no solid black path from~$k$ to~$k'$, then~$k$ cannot be bounded by~$ck'$ for any constant~$c$. Every arc is proved explicitly in Appendix~\ref{sec:relations} (dashed red arcs correspond to unboundedness results), and \emph{all} other relationships are implied.}
  \label{fig:relations}

\end{figure}
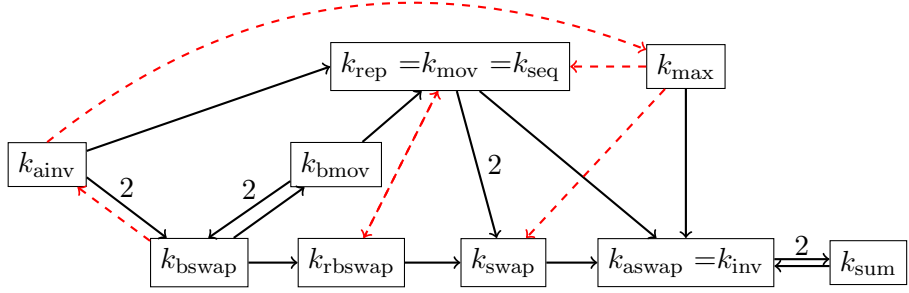

\subsection{Bounded displacement}\label{sub:displacement}

We now consider the two parameters~$\ksum$ and~$\kmax$ that quantify the displacement of elements between~$A$ and~$A^\star$.
More precisely, we define~$\ksum := \sum_{x\in A} |\pos(x) - \rank(x)|$ and~$\kmax := \max_{x\in A} |\pos(x) - \rank(x)|$.
We first derive bounds in terms of~\ksum.

\begin{theorem}
	Every search algorithm needs at least $\lfloor \log (n/2\ksum) \rfloor + 2\ksum + 1$ queries, even if the elements other than~$e$ are in the correct relative order.\footnoteref{foot:trivial_bound}\label{thm:LB_ksum}
\end{theorem}

\begin{proof}
	We give a strategy for an adversary to position the elements of the array adaptively, depending on the queries of the search algorithm. 
	The strategy maintains a range $\{l, \dots, r\}$ of candidate indices for the searched element $e$ that never grows during the course of the strategy. 
	In the beginning, we set $l = 1$ and $r = n$.  
	
	In the first phase of the strategy, we maintain the invariant that all queries to indices $i < l$ yield (and yielded) the result $A[i] < e$, and all queries to indices $j > r$ yield $A[j] > e$. 
	Whenever an index $i \in \{l, \dots, r\}$ is queried, the result depends on whether $\{l,\dots,i\}$ is larger than $\{i,\dots,r\}$ or not. 
	In the former case, the query yields $A[i] > e$ and we set $r = i - 1$. In the latter case, it yields $A[i] < e$ and we set $l = i + 1$. The first phase ends after $\lfloor \log (n/2\ksum) \rfloor - 1 = \lfloor \log (n/4\ksum) \rfloor \leq \lfloor \log (n/(2\ksum+2)) \rfloor$ queries. 
	At this point, we have $r - l + 1 \geq 2\ksum + 2$, hence there are still at least $2\ksum + 2$ positions left that may contain $e$.
	
In the second phase of the adversarial strategy, we answer the next $2\ksum + 1$ queries to indices $i \in \{l, \dots, \lfloor(l+r)/2\rfloor\}$ with $A[i] < e$ and all queries to $i \in \{\lfloor(l+r)/2\rfloor + 1, \dots, r\}$ with $A[i] > e$. 
Afterwards, at least one unqueried index in $\{l,\dots, r\}$ remains. 
It is easy to see that $e$ being in this position is consistent with all queries so far. 
Overall, the position of $e$ cannot be found with fewer than $\lfloor \log (n/2\ksum) \rfloor + 2\ksum + 1$ queries, as claimed. Moreover, all $<$ answers are left of $>$ answers, allowing all elements other than $e$ to be in correct relative order.
\end{proof}

We extract the following corollary from the proof of Theorem~\ref{thm:LB_ksum}.

\begin{corollary}\label{cor:candidate_window}
    There is a constant $c \in \mathbb {N}$, such that for every $l>0$, the adversary can ensure that after $\log n/l + c$ queries, an unqueried subarray of length~$l$ remains, such that all elements to the left of the subarray are smaller than~$e$, while all elements to the right of it are larger than~$e$.
\end{corollary}

We give an algorithm that achieves a optimal number of queries up to an additive up to an additive gap of~$\log\ksum + \Oh(1)$, while being oblivious of the value of~\ksum.

\begin{theorem}
	We can find $e$ obliviously using $\log n/\ksum + 2\ksum + \mathcal O(1)$ queries.\label{thm:UB_ksum}
\end{theorem}

\begin{proof}
	We first perform a regular binary search for $e$, ignoring the fact that we may be misguided by elements being displaced.
    In $\log n + \mathcal{O}(1)$ steps, we find $e$ or an index~$i$ with~$A[i] < e$ and $A[i+1] > e$.
    Let~$\Delta_i := |\pos(A[i]) - \rank(A[i])|$. 
    We have $\rank(e) > \rank(A[i]) \geq i - \Delta_{i}$ and $\rank(e) < \rank(A[i+1]) \leq i+1+\Delta_{i+1}$.
    With~$\pos(e) \in \{\rank(e) - \Delta_e, \dots, \rank(e) + \Delta_e\}$ and~$\Delta_e + \Delta_{i} + \Delta_{i+1} \leq \ksum$, we get
    $$ \pos(e) \in \{ i - \ksum + 1, \dots, i + \ksum \}. $$
    We can search this range obliviously by querying the elements $i, i+1, i-1, i+2, i-2, i+3, \dots$ in this order, until we find $e$.
	During the initial binary search, we already queried $\log \ksum + \mathcal{O}(1)$ of these elements, hence we need a total number of queries equal to
	$$ \log n + 2\ksum - \log \ksum + \mathcal O(1) = \log n/\ksum + 2\ksum + \mathcal O(1). $$
\end{proof}

We now turn our attention to the parameter~\kmax.

\begin{theorem}
    Every search algorithm needs at least $\log n/\kmax + 3\kmax + \mathcal{O}(1)$ queries.\footnoteref{foot:trivial_bound}\label{thm:LB_kmax}
\end{theorem}

\begin{proof}
    By Corollary~\ref{cor:candidate_window}, the adversary can ensure without creating inversions that after using $\log n/4\kmax + \mathcal{O}(1) = \log n/\kmax + \mathcal{O}(1)$ queries the element $e$ may still be at any position of an unqueried subarray of length~$4\kmax$. 
    It is therefore sufficient to show that finding $e$ in an array of length $4\kmax$ and with maximum displacement~$\kmax$ may take~$3\kmax - 2$ queries.
    
    We devise a strategy for the adversary that is split in two phases.
    In the first phase, we maintain that all queried positions in $L=\{1, \dots, 2\kmax\}$ contain elements smaller than~$e$, while the positions in $R=\{2\kmax+1, \dots, 4\kmax\}$ contain elements larger than~$e$.
    The first phase ends when~$\kmax - 1$ positions have been queried in each half of the array.
    
    Without loss of generality, assume that exactly~$\kmax - 1$ positions have been queried in~$L$ at the beginning of the second phase.
    Otherwise, this is true for~$R$ and the argument proceeds analogously.
    We now restrict the position of~$e$ to definitely lie in~$L$.
    The second phase proceeds until another~$\kmax - 1$ positions in~$L$ have been queried.
    All queries to positions in~$R \cup \{1\}$ are answered as before.
    For the queries to positions in~$L \setminus \{1\}$ we return the inverse answer to before.
    
    The third phase proceeds until one more position in~$L$ is queried, which will contain~$e$.
    
    The number of queries up to this point is at least $2(\kmax - 1)$ for the first phase, at least~$\kmax-1$ for the second phase, and at least~$1$ for the third phase.
    Hence, the total number of queries is~$3\kmax-2$ as claimed.
    It remains to argue that we can rearrange the array such that all elements smaller (larger) than $e$ are on the left (right) of $e$ while not moving elements by more than~$\kmax$ positions.
    We fix a final ordering by requiring that the smaller (larger) elements remain in the same relative order.
    
    If position~$1$ was not queried in phase~$2$, then there are~$\kmax-1$ elements smaller than~$e$ and~$\kmax-1$ elements larger than~$e$ in $L$.
    We have~$\mathrm{rank}(e) = \kmax$, thus~$e$ is displaced by at most~$\kmax$ positions.
    All elements smaller than~$e$ are displaced by at most~$\kmax$ positions, since there are at most~$\kmax$ elements that are greater or equal to~$e$ in~$L$.
    Similarly, all elements larger than~$e$ are displaced by at most~$\kmax$ positions, since there are at most~$\kmax - 1$ elements that are smaller or equal to~$e$ in~$L\setminus\{1\}$.
    
    If position~$1$ was queried in phase~$2$, then there are~$\kmax$ elements smaller than~$e$ and~$\kmax-2$ elements larger than~$e$ in $L$.
    The element in position~$1$ has rank 1 and all other elements are displaced by at most~$\kmax$ positions in $L\setminus\{1\}$, as before.
\end{proof}

To obtain a tight upper bound, we need the following observations.

\begin{proposition}\label{proposition:maxinv}
If $A[i]>A[j]$ then $i \geq j - (2\kmax - 1)$.
\end{proposition}

\begin{proof}
If $A[i]>A[j]$ then $\rank(A[i])\geq \rank(A[j])+1$ by definition. Using that $i\geq \rank(A[i])-\kmax$ and $\rank(A[j])\geq j-\kmax$, both by assumption, we derive
\[
i\geq \rank(A[i])-\kmax\geq \rank(A[j])+1-\kmax \geq j-(2\kmax-1),
\]
as claimed.
\end{proof}

\begin{lemma}
	For all $i$ we have $|\{j < i: A[j] > A[i]\}| \leq \kmax$ and symmetrically $|\{j > i: A[j] < A[i]\}| \leq \kmax$.\label{lem:block}
\end{lemma}
\begin{proof}
	For the sake of contradiction, assume that there are $j_1 < j_2 < \dots < j_{\kmax+1} = i - 1$ with $A[j_l] > A[i]$ for all $l \in \{1,2,\dots,\kmax+1\}$ (the symmetrical case can be proven analogously). 
	Let $A[r]$ be such that $\rank(A[r]) = j_{\kmax+1} - \kmax = i - \kmax - 1$.
	We have $r < i$, since $i - \rank(A[r]) = \kmax + 1 > \kmax$.
	Also $A[r] < A[j]$ for all $j \geq i$, since $\rank(A[j]) \geq j - \kmax \geq i - \kmax > \rank(A[r])$.
	On the other hand, the number of elements $A[j]$ with $j<i$ and $A[j] \leq A[i]$ is at most $j_{\kmax+1} - (\kmax+1) = i - \kmax - 2$, since the number of elements $A[j]$ with $j<i$ and $A[j] > A[i]$ is at least $\kmax+1$ by assumption.
	But then, the total number of elements that are smaller than $A[r] < A[i]$ is at most $i - \kmax - 3$.
	This is a contradiction with $\rank(A[r]) = i - \kmax - 1$.
\end{proof}

We now describe an algorithm that achieves an optimal number of queries up to an additive constant, while being oblivious to the value of~\kmax.

\begin{theorem}
	We can find $e$ obliviously using $\log n/\kmax + 3\kmax + \mathcal O(1)$ queries. \label{thm:UB_kmax}
\end{theorem}

\begin{proof}
    We first use a binary search to find~$e$ or an index~$i$ with~$a_i < e < a_{i+1}$ with~$\log n + \Oh(1)$ queries.
    By Proposition~\ref{proposition:maxinv}, we have~$\pos(e) \in W = \{ i - 2\kmax + 1, \dots, i + 2\kmax \}$.
    We query the positions in~$W$, starting from the center (positions $i$ and $i+1$ don't need to be queried again) and moving to the left whenever the number of queried elements of~$W$ larger than~$e$ exceeds the number of smaller elements, and moving to the right otherwise.
    This can be done obliviously, i.e., without knowing~$\kmax$.

    We claim that we are guaranteed to encounter~$e$ within~$3k + 1$ queries.
    To see this, assume without loss of generality that $e$ is in the left half of $W$.
    In this case, we claim that we do not query any elements in $\{i + \kmax + 2, \dots, i + 2\kmax\}$.
    For the sake of contradiction, assume we query positions~$\{i+1, \dots, i + \kmax + 2\}$, i.e., at least $\kmax + 2$ positions in the right half of~$W$.
    This means that we queried at least~$\kmax + 1$ elements smaller than~$e$ in $W$ before finding $e$, by construction of the algorithm.
    But then~$|\{j > \pos(e): A[j] < e\}| > \kmax$, contradicting Lemma~\ref{lem:block}.
    
    We can refine this analysis by observing that we already queried at least $\log k + \Oh(1)$ positions among~$\{i - \kmax + 1, \dots, i + \kmax\}$ during the initial binary search.
\end{proof}


\subsection{Inversions}\label{sub:inversions}

In this section we consider the number of inversions between elements of the array~$A$.
More precisely, we define the number of inversions to be~$\kinv := |\{i < j:A[i] > A[j]\}|$, and the number of adjacent inversions to be $\kainv := |\{i:A[i] > A[i+1]\}|$.

We have~$\ksum \leq \kinv \leq 2\ksum$ (cf.~Proposition~\ref{prop:ksum_kinv}), therefore the results for~$\ksum$ (Theorems~\ref{thm:LB_ksum} and \ref{thm:UB_ksum}) carry over to~\kinv with a gap of~2.

\begin{corollary}
	Every search algorithm needs at least $\log n/\kinv + 2\kinv + \Oh(1)$ queries\footnoteref{foot:trivial_bound}, and we can find~$e$ obliviously with~$\log n/\kinv + 4\kinv + \Oh(1)$ queries.
	\label{cor:LB_kinv}
	\label{cor:UB_kinv}
\end{corollary}

In general, we cannot hope to obtain results of similar quality for the smaller parameter~\kainv. In fact, already for $\kainv=1$ any search algorithm needs to query all $n$ elements.

\begin{proposition}
For $\kainv \geq 1$, no algorithm can find $e$ with less than $n$ queries.
\end{proposition}

\begin{proof}
Consider the family of arrays that are obtained from $[1,\ldots,n]$ by moving $n$ to an arbitrary position (possibly leaving it in place); all arrays of this form have $\kainv\leq 1$ since the only possible adjacent inversion is between $n$ and the succeeding element). An adversary may use this family to force any search algorithm to query all $n$ positions when searching for $e=n$: The adversary will answer the first $n-1$ queries by $<$, maintaining that all arrays where $n$ is in any unqueried position are consistent with the answers given so far. (This is easy to see since all other elements are smaller than $e=n$.)
\end{proof}

Fortunately, we can do much better if the target $e$ is guaranteed to be in the correct position relative to sorted order, i.e., if $\pos(e)=\rank(e)$. Note that this restriction still allows us to prove a lower bound on the necessary number of queries that is much larger than all preceding results. We will complement this lower bound by a search algorithm that matches it tightly (up to lower-order terms). Both upper and lower bound hinge on the question of how efficiently (in terms of queries) an algorithm can find a good estimate of $\rank(e)$ by querying the array.

\begin{theorem}\label{theorem:app:kaiv:lower}\label{thm:LB_kainv}
Every search algorithm needs at least $2\sqrt{2n\kainv}-o(\sqrt{n\kainv})$ queries, even if $\pos(e) = \rank(e)$.\footnoteref{foot:trivial_bound}
\end{theorem}

\begin{proof}
We describe an adversary that will force any search algorithm to use at least the claimed number of queries, i.e., at least $2\sqrt{2n\kainv}$ minus lower order terms. The adversary will not fix the actual contents of the array beforehand but will guarantee throughout that there exists an $n$-element array with at most $\kainv$ adjacent inversions that is consistent with all queries answered so far. We will consider positions of the underlying array to be numbered $1,\ldots,n$, and assume that $n$ is even for convenience.

At high level, the adversary aims to place the target $e$ as close to the middle of the array as possible. Accordingly, his \emph{standard response} to queries to the first half is $<$ whereas it is $>$ for the second half. Eventually, he will be forced to pick a concrete array and position of $e$ that is consistent with all queries. If $e$ is placed in position $x$ in the first half then previous queries between $x$ and the middle have identified elements that are smaller than the target, and that are now found to the right of it. To adhere to the restriction that $\pos(e)=\rank(e)$, the adversary needs to choose an array that has the same number of larger elements to the left of $e$. He needs to place blocks of such elements between positions left of $x$ that are already queried (and contain smaller elements) without causing too many adjacent inversions; these blocks are called \emph{hidden blocks} in reference to the fact that none of their positions has been queried before. (All of this is symmetric for placement in the second half.)

Let us now give a detailed description of the adversary's strategy. The standard response is $<$ for positions $1,\ldots,\frac{n}{2}$ and $>$ for positions $\frac{n}{2}+1,\ldots,n$. The adversary keeps track of the following values: $p\geq 0$ is the smallest value such that $\frac{n}{2}-p$ has not been queried yet; $q\geq 0$ is the smallest value such that $\frac{n}{2}+1+q$ has not been queried yet; $\ell$ is the number of queries that have been made to positions $1,\ldots,\frac{n}{2}-p-1$; and $r$ is the number of queries that have been made to positions $\frac{n}{2}+q+2,\ldots,n-1$. Initially we have $p=q=\ell=r=0$. Observe that $p$ and $q$ will never decrease, but $\ell$ and $r$ may decrease upon increases of $p$ or $q$, respectively. Note that $p+\ell$ and $q+r$ never decrease: E.g., $p+\ell$ counts the queries on $1,\ldots,\frac{n}{2}-p-1$ and those on $\frac{n}{2}-p+1,\ldots,\frac{n}{2}$; there is never any previous query for position $\frac{n}{2}-p$ by choice of $p$. Thus, we also see that $p+q+\ell+r$ is always equal to the total number of queries made so far by the search algorithm.

Concretely, the adversary plans to put the target~$e$ either in position $\frac{n}{2}-p$ in the first half or in position $\frac{n}{2}+1+q$ in the second half of the array. Queries to the first, respectively second, half of the array may force him to abandon the corresponding option. Abandoning the first, say not putting the target in the first half, allows him to continue answering $<$ in the first half, and to answer $>$ in the second half until he needs to commit to a position $\frac{n}{2}+1+q$ in the second half. In other words, the standard response can be continued. Once a standard response would force to abandon the second option he instead needs to commit to an instantiation that puts the target in the previously still feasible half. At this point, no further queries will be guaranteed, since the current query could in principle be to the position of $e$ in the chosen instantiation. Nevertheless, we show that the number $p+q+\ell+r$ of queries of the search algorithm invested will be sufficiently large at this point. To do so, we give a lower bound on $p+\ell$ for the point when the first half is no longer feasible, and an analog lower bound on $q+r$ for the second half.

Intuitively, the adversary aims to have the target as close as possible to the center of the array. Of course, queries to positions around $\frac{n}{2}$ will eventually increase the values of $p$ and $q$, which are bounds for how close the target $e$ can be to the center. Since the adversary needs to fulfill $\pos(e)=\rank(e)$, the position of $e$ forces the adversary to choose an array with the correct numbers of smaller and larger elements. In particular, if the adversary chooses for example to place the target $e$ in the current position $p$, then all elements in positions $p+1,\ldots,\frac{n}{2}$ are smaller than the target (due to queries that were already answered); these are exactly $p$ elements. Thus, to balance out the numbers the adversary needs to choose an array such that at least $p$ elements larger than the target are in positions $1,\ldots,p-1$ (we will go for exactly $p$ and have no additional smaller elements succeeding the target). This is hindered by the fact that $\ell$ queries were already made on this part, and the fact that every maximal block of larger elements in the first half necessarily ends with an adjacent inversion (either with a smaller element or with the target). This will eventually force the adversary to ``give up'' on one half of the array or (when this happens the second time) to pick a concrete instantiation with the target placed in the second half.

Now, assume that some query is made by the search algorithm. We will discuss in detail what happens for a query to a position in $1,\ldots,\frac{n}{2}$; queries to $\frac{n}{2}+1,\ldots,n$ are treated symmetrically. For a query to a position in $1,\ldots,\frac{n}{2}$ the adversary tests whether he can answer with $<$ and still maintain existence of a consistent instantiation that places $e$ in the first half. To this end, pretend that the query is answered $<$, update $p$ and $\ell$, and check whether there is a consistent instantiation with target in position $\frac{n}{2}-p$ that has at most $\kainv$ adjacent inversions. Note that the adversary does not need to do this optimally. It suffices to have a strategy that causes the claimed number of queries, and he may give up even though there could still be a consistent instantiation.

From the perspective of a sorted array, placing the target in position $\frac{n}{2}-p\in\{1,\ldots,\frac{n}{2}\}$ only conflicts with the queries to positions $\frac{n}{2}-p+1,\ldots,\frac{n}{2}$ that were already answered with $<$, of which there are exactly $p$. To balance out the total numbers of larger and smaller elements, the adversary checks for an instantiation that puts $k$ blocks of larger elements into positions $1,\ldots,\frac{n}{2}-p-1$, surrounded by smaller elements. This causes $k$ adjacent inversions between the last element of a block and the subsequent element. Moreover, one adjacent inversion exists between the target in position $\frac{n}{2}-p$ and its successor, which must be a smaller element due to queries, barring the trivial case of $p=0$ where a fully sorted array without adjacent inversions suffices. Thus, the adversary may choose $k:=\kainv-1$. In the interest of recycling the argument later, we will do the analysis in terms of $k$ and only plug in $\kainv-1$ at the very end.

Thus, it suffices to check whether there are $k$ non-overlapping blocks in positions $1,\ldots,\frac{n}{2}-p-1$ of elements that have not been queried yet, whose total size is at least $p$; placing $p$ larger elements in such blocks will be consistent with queries answered so far. If such a choice of blocks exist then a valid instantiation would be to place $p$ larger elements in these blocks, surrounded by smaller ones, followed by the target in position $\frac{n}{2}-p$, followed by $p$ smaller elements, and finally by $\frac{n}{2}$ larger elements. Observe that the adversary can choose the smaller respectively larger elements freely and, hence, there are no adjacent inversions inside blocks of smaller respectively larger elements. Thus, if the $p$ larger elements that are required can be placed in $k$ blocks in positions $1,\ldots,\frac{n}{2}-p-1$ then there is an instantiation with at most $k+1$ adjacent inversions ($k$ from the blocks and one between the target and its successor). With $k\leq\kainv-1$ this suffices to maintain the invariant of an existing consistent instantiation.

Let us now derive a lower bound for $p+\ell$ for the case that the adversary cannot find $k$ non-overlapping blocks of unqueried elements in $0,\ldots,\frac{n}{2}-p-1$. Clearly, the total size of \emph{any} $k$ maximal unqueried blocks in $0,\ldots,\frac{n}{2}-p-1$ must be less than $p$. We can relate $n$, $p$, and $k$ by counting the number $F$ of unqueried elements in $1,\ldots,\frac{n}{2}-p-1$ in two different ways. Clearly, because $\ell$ is the number of queries made to $1,\ldots,\frac{n}{2}-p-1$ we have
\begin{align}
F=\frac{n}{2}-p-1-\ell.\label{equation:kaiv:lower:f}
\end{align}
On the other hand, letting $B_1,\ldots,B_s$ denote the maximal unqueried blocks in $1,\ldots,\frac{n}{2}-p-1$, we have
\[
F=\sum_{i=1}^s |B_i|.
\]
For convenience, we explicitly include empty blocks between adjacent queried positions, before a position $1$ if it is queried, and after position $\frac{n}{2}-p-1$ if it is queried; in this way, the number of maximal unqueried blocks is exactly $\ell+1$. Since any $k$ of these blocks have total size less than $p$ we get $|B_1|+\ldots+|B_k|<p$, $|B_2|+\ldots+|B_{k+1}|<p$, and so on (wrapping around indices larger than $s$). Summing up these $s=\ell+1$ inequalities we get
\begin{align*}
(\ell+1)\cdot p= s\cdot p > k\cdot \sum_{i=1}^s |B_i| = k\cdot F.
\end{align*}
Together with $(\ref{equation:kaiv:lower:f})$ this yields
\begin{align*}
(\ell+1)\cdot p> k\cdot F=k\cdot \left( \frac{n}{2} -p -1 -\ell \right).
\end{align*}
We bring this inequality into a more convenient form to derive a lower bound for $\ell+p$:
\begin{align}
&&(\ell+1)\cdot p &{}> k\cdot \left( \frac{n}{2} -p -1 -\ell \right)&&\nonumber\\
\Leftrightarrow && (\ell+1)\cdot p +k\cdot (p + \ell) &{}> k\cdot \left( \frac{n}{2}-1 \right)&& \label{equation:kaiv:lower:main}
\end{align}
The left hand side is upper bounded by
\begin{align}
\left( \frac{p+\ell+1}{2} \right)^2 + k\cdot (p+\ell)\geq (\ell+1)\cdot p +k\cdot (p + \ell),\label{equation:kaiv:lower:lhs}
\end{align}
since, for any $x,y\geq 0$ we have
\begin{align*}
&& \left( \frac{x-y}{2} \right)^2 & \geq 0\\
\Rightarrow && \frac{x^2}{4} - \frac{xy}{2} + \frac{y^2}{4} & \geq 0\\
\Rightarrow && \frac{x^2}{4} + \frac{xy}{2} + \frac{y^2}{4} & \geq xy\\
\Rightarrow && \left( \frac{x+y}{2} \right)^2 & \geq xy;
\end{align*}
we use it with $x=\ell+1$ and $y=p$ to obtain $(\ref{equation:kaiv:lower:lhs})$.

Combining $(\ref{equation:kaiv:lower:main})$ with $(\ref{equation:kaiv:lower:lhs})$ yields
\begin{align}
&&\left( \frac{p+\ell+1}{2} \right)^2 + k\cdot (p+\ell){} &{} > k\cdot \left( \frac{n}{2}-1 \right),&&\nonumber\\
\Leftrightarrow && \left( \frac{p+\ell+1}{2} \right)^2 + k\cdot (p+\ell+1) - k - k\cdot \left( \frac{n}{2}-1 \right){} &{} >0,&& \nonumber\\
\Leftrightarrow && \left( \frac{p+\ell+1}{2} \right)^2 + k\cdot (p+\ell+1) - k\cdot \frac{n}{2} &{} >0,&&\\
\Leftrightarrow && \left( p+\ell+1 \right)^2 + 4k\cdot (p+\ell+1) - 2k n &{} >0.&&\label{equation:kaiv:lower:quadratic}
\end{align}
Using $p+\ell+1>1$ and computing the roots of \eqref{equation:kaiv:lower:quadratic} wrt.\ $p+\ell+1$ we arrive at
\begin{align*}
&&p+\ell+1 &{} > - 2k + \sqrt{4k^2+2kn}.&&\\
\Rightarrow && p+\ell &{} \geq \sqrt{4k^2+2kn} -2k.&&
\end{align*}
Thus, if the proposed instantiation is not possible using $k$ blocks of larger elements then we have
\[
p+\ell\geq \sqrt{4k^2+2kn} -2k=\sqrt{2kn}-o(\sqrt{kn}). 
\]

Similarly, if there is no feasible instantiation placing the target at $\frac{n}{2}+1+q$ in the second half of the array then we can prove that $q+r\geq \sqrt{4k^2+2kn}-2k=\sqrt{2kn}-o(\sqrt{kn})$. 
We give the calculations here for completeness.

Let us check first that the adversary can use the same number $k=\kainv-1$ of hidden blocks: His goal is let the second half, i.e., positions $\frac{n}{2}+1,\ldots,n$, contain (in order) $q$ larger elements, the target in position $\frac{n}{2}+1+q$, and larger elements interspersed by up to $k$ blocks of smaller elements. As before, elements within the group of larger/smaller elements can be assumed to be sorted, thus, adjacent inversions are only possible before the target (if preceded by a larger element), or before a smaller element (if preceded by the target or an element larger than the target). Clearly, we get an adjacent inversion at each of the $k$ blocks, a single inversion between the $q$ larger elements and the target, and no further adjacent inversions. Hence, $k=\kainv-1$ hidden blocks are a feasible choice.

If placing at $\frac{n}{2}+1+q$ is infeasible then in particular there are no $k$ blocks of unqueried elements in $\frac{n}{2}+q+2,\ldots,n$ of total size at least $q$. Again, using that the number $F'$ of unqueried elements in $\frac{n}{2}+q+2,\ldots,n$ is equal to $\frac{n}{2}-q-1-r$, but also equal to the total size of the $s'=r+1$ maximal unqueried blocks $B'_1,\ldots,B'_{s'}$ (including size-zero blocks), we get
\[
(r+1)\cdot q=s'\cdot q > k\cdot \sum_{i=1}^{s'} |B'_i| = k\cdot F'= k\cdot \left(\frac{n}{2}-q-1-r\right).
\]
This implies
\[
(r+1)\cdot q + k \cdot (q+r) > k \cdot \left(\frac{n}{2}-1\right).
\]
At this point it is obvious that we arrive at the same lower bound for $q+r$ as we had for $p+\ell$, i.e., as claimed above.

Thus, we conclude that if the adversary gets to use instantiations with $k$ hidden blocks (as above) then he can enforce a total of at least
\[
p+\ell+q+r\geq 2 \sqrt{4k^2+2kn}-4k = 2 \sqrt{2kn} - o(\sqrt{kn})
\]
queries. For the case of $k=\kainv-1$ we get a lower bound of $2\sqrt{2\kainv n}-o(\sqrt{\kainv n})$ as claimed.
\end{proof}

We now describe an algorithm that achieves the optimal number of queries (up to lower order terms).
Note that the algorithm requires knowledge of~\kainv.

\begin{theorem}\label{theorem:app:kaiv:upper}
We can find $e$ using $2\sqrt{2n\kainv}+o(\sqrt{n\kainv})$ queries if $\pos(e) = \rank(e)$.
\label{thm:UB_kainv}
\end{theorem}

\begin{proof}
For the description of our algorithm it will be convenient to take the array as having positions numbered $0,\ldots,n-1$. As before, a query for some position $i$ will yield $<$, $>$, or $=$ depending on whether $A[i]<e$, $A[i]>e$, or $A[i]=e$. The assumption that $\pos(e)=\rank(e)$ will be crucial; the algorithm will attempt to get a good estimate for $\rank(e)$ and then query a certain range around the estimated position.

The algorithm first fixes a block size of the form $p=c\cdot \sqrt{\frac{n}{\kainv}}$ for some constant~$c$ that we will fix later, and then queries positions $0, p+1, 2(p+1), \ldots, \lceil\frac{n-1}{p}\rceil(p+1),n-1$. 
We refer to these positions as the \emph{grid}. 
In this way, unqueried blocks of size (at most) $p$ remain, with every block sandwiched between two grid positions. We will refer to these blocks according to query outcomes to the adjacent grid positions: $<>$-blocks, $<<$-blocks, $><$-blocks, and $>>$-blocks. We use $\sharp(xy)$ to denote the number of $xy$-blocks for $x,y\in\{<,>\}$. We assume that the target $e$ is not at a grid position, since otherwise the claimed bound holds trivially.

Intuitively, since there are at most~\kainv adjacent inversions, there can only be a limited number of $>>$-blocks containing elements smaller than $e$ and of $<<$-blocks containing larger elements: Either constellation leads to at least one adjacent inversion inside the block. Moreover, every $><$-block must contain an adjacent inversion, which upper bounds their number by $\kainv$ (we will give a better bound later). The number of $<>$-blocks is at most equal to the number of $><$-blocks plus one, because there must be an $><$-block somewhere between any two $<>$-blocks.

The algorithm now tries to estimate the position of $e$ in the array. Crucially, because $\pos(e) = \rank(e)$, there must be exactly $\pos(e)$ elements smaller than~$e$ in the array. We denote by $q(x)$ the number of queries to grid positions that returned $x$, for $x\in\{<,>,=\}$. We further denote by $\eta(xy)$, for $x,y\in\{<,>\}$, the number of adjacent inversions involving an element of a $xy$-blocks, including adjacent inversions between an element of the block and an adjacent grid position. We denote by $N(xy)$, for $x,y\in\{<,>\}$, the total number of positions in $xy$-blocks (not counting the queries adjacent to the block).

Observe that~$\pos(e)$ is equal to the number of elements smaller than~$e$ in the array.
We refer to these elements simply as \emph{small} elements, as opposed to \emph{large} elements that are larger than~$e$.
We want to bound the number of small elements in order to obtain a range of positions that our algorithm needs to search.
By definition, we have~$q(<)$ small elements in grid positions.
We now give upper and lower bounds for the number of small elements in each type of block in terms of $\sharp(xy)$ and $\eta(xy)$. It can be observed that upper and lower bound are attained by blocks not containing $e$, hence we will tacitly ignore this case.

\begin{description}
	
 \item[$<<$-blocks:] The maximum number~$N(<<)$ of smaller elements in $<<$-blocks is attained if there are no adjacent inversions. If there is at least one adjacent inversion in a $<<$-block, it can have anywhere between $0$ and block size (at most $p$) small elements: If the large elements form a single block then there is exactly one adjacent inversion between the last element of the block and the succeeding element (possibly a grid position). Overall, the number of small elements in $<<$-blocks is between $N(<<)-\eta(<<)\cdot p$ and $N(<<)$.

 \item[$>>$-blocks:] The minimum number of small elements in $>>$-blocks is attained if there are no adjacent inversions; in this case there are no smaller elements in these blocks. A $>>$-block with at least one adjacent inversion can contain any number between $0$ and the block size (at most $p$) of small elements: A consecutive block of small elements causes a single adjacent inversion at its start. Overall, the number of small elements in $>>$-blocks is between $0$ and $\eta(>>)\cdot p$.
 
 \item[$><$-blocks:] Every $><$-block contains at least one adjacent inversion. In such a block, even without any further adjacent inversions, we can have between $0$ and the block size (at most $p$) smaller elements: The block may contain all large elements followed by all small ones, and have a single adjacent inversion between the last large and first small element. Overall, the number of small elements in $><$-blocks is between $0$ and $\sharp(><)\cdot p$. Since there is at least one adjacent inversion per $><$-block we have $\sharp(><)\leq\eta(><)$ and the upper bound becomes $\eta(><)\cdot p$.
 
 \item[$<>$-blocks:] In $<>$-blocks we can have any number of small elements followed by large ones (with the total being at most the block size $p$) without any adjacent inversions. We noted already that the number of these blocks is at most $\sharp(><)+1$. Overall, the number of small elements in $<>$-blocks is between $0$ and $\sharp(<>)\cdot p\leq(\sharp(><)+1)\cdot p\leq(\eta(><)+1)\cdot p$.
 
\end{description}

In total, we get that there are at least
\[
q(<) + N(<<) - \eta(<<)\cdot p
\]
and at most
\[
q(<) + N(<<) + \eta(>>)\cdot p + \eta(><)\cdot p + (\eta(><)+1)\cdot p
\]
elements that are smaller than~$e$.

Since~$\pos(e)$ is equal to the number of small elements in the array, the gap (plus one) between these two bounds is the number of positions that the algorithm has to query in order to find $e$ or be sure that it is not present. The gap (difference) is upper-bounded by
\begin{align*}
& \left( q(<) + N(<<) + \eta(>>)\cdot p + \eta(><)\cdot p + (\eta(><)+1)\cdot p \right)\\
& - \left( q(<) + N(<<) - \eta(<<)\cdot p \right) + 1\\
={} & \eta(>>)\cdot p + \eta(><)\cdot p + (\eta(><)+1)\cdot p + \eta(<<)\cdot p +1\\
={} & \eta(>>)\cdot p + \eta(<<)\cdot p + 2\eta(><)\cdot p + p +1.
\end{align*}

Since $\eta(>>)+\eta(<<)+\eta(><)+\eta(<>)\leq \kainv$ and since all values are non-negative, this expression is maximized for $\eta(><) = \kainv$ and $\eta(<>) = \eta(<<) = \eta(>>) = 0$, i.e., if there are no adjacent inversions inside $<<$-blocks, $>>$-blocks, and $<>$-blocks. 
We get a range of at most $(2\kainv+1)\cdot p + 1$ positions to search, which coincides with the claim of the theorem, since~$p = c\sqrt{n/\kainv}$.
Unfortunately, while the algorithm knows~$\kainv$ and can thus compute the \emph{number of elements} in the range it has to search, it does not know exactly \emph{where the range starts or ends} because it does not have access to the values of~$\eta$ for each block and cannot compute the value of the upper or lower bound. 

Since the search range is maximized when all inversions fall in $><$-blocks it makes sense to refine our initial grid in order to get rid of $><$-blocks altogether.
We can do this by running a binary search on each $><$-block to find an adjacent inversion in at most~$1+\log p$ queries.
To do this, we simply query the center element and recurse on the left subblock if it is small and on the right subblock if it is large, until we are left with a subblock containing a large element followed by a small one. 
By extending our initial grid by the additional query positions, we replace each $><$-block by some number of $>>$-blocks, an empty $><$-block (the adjacent inversion), and some number of $<<$-blocks.
We update the values of $\sharp$, $N$, $\eta$ and $q$ accordingly.

Overall, our algorithm spends at most $\sharp(><)\cdot(1+\log p)$ queries on the refinement of the grid, resulting in all $><$-blocks being empty. 
Thus, we get tighter bounds for the candidate range for~$\pos(e)$ by setting the block size of $><$-blocks to 0 (instead of $p$), which eliminates the term $\eta(><)\cdot p$.
Note that while we introduced additional $<<$-blocks and $>>$-blocks and may thus have increased~$\eta(<<)$ and~$\eta(>>)$, the additional blocks contain fewer than~$p$ elements.
Nevertheless, we can use the generous bound of~$p$ for the size of all blocks, since the block size only appears negatively in the lower bound and positively in the upper bound. 
We retain the lower bound of
\[
q(<) + N(<<) - \eta(<<)\cdot p
\]
and get a new upper bound of
\[
q(<) + N(<<) + \eta(>>)\cdot p + (\eta(><)+1)\cdot p,
\]

Using that $-\kainv \leq - \eta(<<) \leq 0$ and $0 \leq \eta(>>) + \eta(><) \leq \kainv$ we conclude that the number of smaller elements is lower bounded by
\[
q(<) + N(<<) - \kainv \cdot p 
\]
and upper bounded by
\[
q(<) + N(<<) + \kainv \cdot p + p.
\]

Since the bounds depend only on values that are known to the algorithm, it can simply query all positions in this range. Since $\pos(e)=\rank(e)$, if the target $e$ is contained in the array then it must be in the position that equals the number of smaller elements. Thus, it suffices to query the $2\kainv p + p + 1$ elements between the above bounds. 
Note that small savings are possible here because some of the positions have been queried previously, but we will not analyze this.

Thus, the overall number of queries needed to establish the initial grid setup, for its refinement, and for the final sweep of the candidate range for~$\pos(e)$ is at most
\begin{align}
\frac{n}{p+1} + 2 + \sharp(><)\cdot(1+\log p) + 2\kainv\cdot p + p + 1 = \frac{n}{p+1} + 2\kainv\cdot p + o(\kainv p),\label{term:adjinv:costupperbound}
\end{align}

where $\frac{n}{p+1} +2$ upper bounds the number queries needed to establish the initial grid that partitions the array into unqueried blocks of size at most $p$ each. 
Rounding up slightly, we are left with choosing $c$ in $p=c\cdot \sqrt{\frac{n}{\kainv}}$ in order to minimize
\begin{align}
\frac{n}{p+1} + 2\kainv\cdot p\leq \frac{n}{p} + 2\kainv\cdot p. \label{term:adjinv:simplifiedcost}
\end{align}

In other words, after plugging in $p=c\cdot \sqrt{\frac{n}{\kainv}}$, we seek $c$ that minimizes
\[
\frac{n}{c\cdot \sqrt{\frac{n}{\kainv}}} + 2\kainv\cdot c\cdot \sqrt{\frac{n}{\kainv}}=\frac{1}{c}{\sqrt{n\kainv}} + 2 c\cdot \sqrt{n\kainv}=\left(\frac{1}{c}+2c \right) \cdot \sqrt{n\kainv}.
\]
Choosing $c=\frac{1}{\sqrt{2}}$ yields $(\frac{1}{c}+2c)=2\sqrt{2}$ and the claimed bound of $2\sqrt{2n\kainv}+o(\sqrt{n\kainv})$ queries for finding the target $e$.
\end{proof}


\subsection{Edit distances}\label{sub:edit_distances}

In this section, we consider parameters that bound the number of elementary array modifications needed to sort the given array~$A$.
More precisely, a \emph{replacement} is the operation of replacing one element with a new element, and we let \krep be the (minimum) number of replacements needed to obtain a sorted array. 
A \emph{swap} is the exchange of the content of two array positions, and \kswap is the number of swaps needed to sort~$A$.
We let \kaswap be the number of swaps of pairs of neighboring elements needed to sort~$A$.
A \emph{move} is the operation of removing an element and re-inserting it after a given position~$i$, shifting all elements between old and new position by one.
We let \kmov be the number of moves needed to sort~$A$.

Clearly, starting from a sorted array, we can move~$e$ to any position, without using more than a single move or swap, or two replacements involving~$e$.
To find~$e$ we then have to query the entire array.

\begin{proposition}
For $\kmov\geq1$, $\kswap\geq1$, or $\krep\geq2$, no algorithm can find $e$ with less than $n$ queries in general.
\end{proposition}

We can obtain significantly improved bounds if the element~$e$ remains at its correct position relative to the sorted array.
Recall that we can interpret~\kfaults as a measure of disorder via $\kfaults(e) = |{i: (i<\pos(e) \wedge A[i] > e) \vee (i>\pos(e) \wedge A[i] < e)}|$.

\begin{lemma}
If $\mathrm{rank}(e)=\mathrm{pos}(e)$, then $\kfaults(e)\leq \min\{2\krep, 4\kswap, 2\kmov\}$.
\end{lemma}

\begin{proof}
Consider an array $A$ with $\krep=k$. 
We get a modified array $A'$ with $\krep=k'$ by switching out all elements smaller than $e$ in $A$ with a common element $e_{<}<e$ and all larger elements by $e_{>}>e$.
Let $r_{j}=(i_{j},e_{j}),j\in\{1,\dots,k\}$ be $k$ replacements that transform $A$ into a sorted array. We define $r_{j}'=(i_{j},0)$ if $e_{j}<e$ and $r_{j}'=(i_{j},2e)$ otherwise. 
Clearly, the replacements $r_{j}'=(i_{j},e_{j}),j\in\{1,\dots,k\}$ transform $A'$ into a sorted array, and, hence, $k'\leq k$. 
Let $m$ be the number of entries left of $e$ that contain $2e$. 
Since $\mathrm{rank}(e)=\mathrm{pos}(e)$, we have that $m$ is also the number of entries equal $0$ right of $e$. 
It is clear that $k'\geq m$, since we have $m$ disjoint pairs of elements in the wrong relative order that need to be repaired by replacing at least one of the two. 
On the other hand, in both $A$ and $A'$, we have $\kfaults(e)=2m\leq2k'\leq2k$.

Now assume $A$ has $\kswap=k$ and let $s_{j}=(i_{j},i_{j}'),j\in\{1,\dots k\}$
be the $k$ swaps (given by the indices of the swapped elements) that
transform $A$ into a sorted array. Clearly, the same sequence of
swaps turns $A'$ into a sorted array, and, hence, $k'\leq k$, where
$A'$ has $\kswap=k'$. As before, we have $m$ disjoint pairs
of elements in the wrong relative order that need to be repaired by
switching at least one of the two elements. Each switch can repair
two such pairs, and we thus have $k'\geq m/2$. In both $A$ and $A'$,
we have $\kfaults(e)=2m\leq4k'\leq4k$.

Similarly, the transformation from~$A$ to~$A'$ does not increase~$\kmov$.
Again, we have~$m$ disjoint pairs of elements in the wrong relative order that need to be repaired by moving at least one of the two.
Hence, $\kfaults(e)=2m\leq2\kmov$.
\end{proof}

With this lemma, we can translate the upper bounds of any algorithm for~\klies.
Before we do, we introduce another measure of disorder, that turns out to be closely related to \kmov.
We define the parameter \kseq to be such that $n-\kseq$ is the length of a longest nondecreasing subsequence in~$A$.
It turns out that~$\kmov = \kseq$ (cf.~Proposition~\ref{prop:kseq_kmov}), and we can thus include this parameter in our upper bound.

\begin{theorem}
Let $f\colon\mathbb{N}^{2}\to\mathbb{N}$. If $\mathrm{rank}(e)=\mathrm{pos}(e)$,
and we can find $e$ with $f(n,\klies)$ queries, then we can find
$e$ obliviously with $\min\{f(n,2\krep),f(n,4\kswap), f(n,2\kmov), f(n,2\kseq)\}$ queries.
\label{thm:UB_krep}
\label{thm:UB_kswap}
\label{thm:UB_kmov}
\label{thm:UB_kseq}
\end{theorem}

We can also carry over the lower bound from Corollary~\ref{cor:LB_kfaults}.

\begin{corollary}
	For every $c\in\mathbb{N}$ and~$\pos(e) = \rank(e)$, no algorithm operating on the search tree can find $e$ with less than $\log n+ck$ queries in general, for~$k\in\{\krep, \kswap, \kmov, \kseq\}$.\footnoteref{foot:trivial_bound}
	\label{cor:LB_krep}
	\label{cor:LB_kswap}
	\label{cor:LB_kmov}
	\label{cor:LB_kseq}
\end{corollary}

\begin{proof}
Because of~$\pos(e) = \rank(e)$, we group all indices in~$A$ with at least one wrong query into~$m$ disjoint pairs with one element left of~$e$ and one element right of~$e$ in each pair.
Starting from a sorted array, we can clearly produce the array~$A$ by swapping each pair.
This requires~$\kfaults/2$ swaps, or $\kfaults$ element moves or replacements.
We get~$\kfaults \geq \max\{\krep, \kswap, \kmov\}$.
Corollary~\ref{cor:LB_kfaults} together with~$\kmov = \kseq$ (Proposition~\ref{prop:kseq_kmov}) thus implies the claim.
\end{proof}

Finally, we immediately obtain bounds for \kaswap from Corollary~\ref{cor:LB_kinv}, because $\kaswap = \kinv$ (cf.~Proposition~\ref{prop:kaswap_kinv}).

\begin{corollary}
	Every search algorithm needs at least $\log n/\kaswap + 2\kaswap + \Oh(1)$ queries\footnoteref{foot:trivial_bound}, and we can find~$e$ obliviously with~$\log n/\kaswap + 4\kaswap + \Oh(1)$ queries.
	\label{cor:LB_kaswap}
	\label{cor:UB_kaswap}
\end{corollary}


\subsection{Block edit distances}\label{sub:block_distances}

In this section we consider the parameters $\kbswap$, $\krbswap$, and $\kbmov$, which bound the number of block edit operations needed to sort~$A$. A \emph{block} is defined to be a subarray~$A[i, i+1, \dots,j]$ of consecutive elements. A \emph{block swap} is the operation of exchanging a subarray~$A[i,\dots,j]$ with a subarray~$A[i',\dots,j']$ and vice versa, where~$i < j < i' < j'$. Note that a block swap may affect the positions of other elements in case that the two blocks are of different sizes. The parameter~$\kbswap$ bounds the number of block swaps needed to sort~$A$. For~$\krbswap$ we only allow block swaps restricted to pairs of blocks of equal sizes. Finally, for~$\kbmov$ one of the two blocks must be empty, i.e., only block moves are allowed. For all three parameters one can easily prove that, without further restrictions, search algorithms need to query all positions of an array to find the target.

\begin{proposition}
For $\kbswap \geq 1$, $\krbswap \geq 1$, or~$\kbmov \geq 1$, no algorithm can find $e$ with less than $n$ queries.
\end{proposition}

\begin{proof}
For $\kbmov$ and $\kbswap$ consider the family of arrays obtained from $[1,\ldots,n]$ by moving~$n$ to an arbitrary position; arrays of this form have $\kbmov=\kbswap\leq 1$. An adversary may answer the first $n-1$ queries by $<$ while maintaining that placing $e=n$ in any of the unqueried positions is consistent with all given answers. 

For $\krbswap$ consider the family of arrays obtained from $[1,\ldots,n]$ by swapping $n$ with any element (or keeping it in place; such arrays have $\krbswap\leq 1$. An adversary may again answer the first $n-1$ queries by $<$ while maintaining that placing $e=n$ in any unqueried position is consistent with the given answers.
\end{proof}

Complementing this lower bound, for all of three parameters, an upper bound of $\Oh(\sqrt{nk})$ for finding $e$ when $\pos(e)=\rank(e)$ follows immediately from the results for $\kainv$ of Section~\ref{sub:inversions} and the fact that $\kainv\leq 2 \kbswap$ (Proposition~\ref{prop:kainv_kbswap}) and $\kbswap \leq \min\{\krbswap, \kbmov\}$ (Propositions~\ref{prop:kbswap_krbswap} and \ref{prop:kbswap_kbmov}). By inspecting the upper and lower bounds proved for $\kainv$, and adapting the proofs, we are able to obtain tight leading constants in the upper and lower bounds for $\kbmov$ and $\kbswap$, and leading constants within a factor of $\sqrt{2}$ for $\krbswap$. First, we adapt the lower bound for $\kainv$ (Theorem~\ref{theorem:app:kaiv:lower}) to $\krbswap$ and $\kbswap$.

\begin{theorem}\label{theorem:kbsr:lower}
Every search algorithm needs at least $2\sqrt{2n\krbswap}-o(\sqrt{n\krbswap})$ queries to find $e$, even if $\pos(e)=\rank(e)$.\footnoteref{foot:trivial_bound}
\label{thm:LB_krbswap}
\end{theorem}

\begin{proof}
We use the same adversary setup as in Theorem~\ref{theorem:app:kaiv:lower} but we need to now make sure that the adversary maintains existence of a suitable instantiation that is at most $\krbswap$ restricted block swaps away from being sorted. We will show that the same strategy can be used, but with $k=\krbswap$ hidden blocks. 

We only discuss the case that the adversary commits to an instantiation with the target $e$ in the first half of the array and having $k$ blocks of large elements between already queried positions. By the analysis from the proof of Theorem~\ref{theorem:app:kaiv:lower} we know that when the adversary commits to putting the target into position $\frac{n}{2}-p$ in the first half, there exist $k$ non-overlapping blocks of unqueried elements in the first half, and with total size at least $p$. The instantiation consists of (in this order) small elements in positions $1,\ldots,\frac{n}{2}-p-1$, interspersed with $k$ blocks of large elements of total length $p$; the target $e$ in position $\frac{n}{2}-p$; $p$ small elements in positions $\frac{n}{2}-p+1,\ldots,\frac{n}{2}$; and $\frac{n}{2}$ large elements in positions $\frac{n}{2}+1,\ldots,n$. Clearly, this can be turned by $k=\krbswap$ restricted block swaps into an array with small elements, followed by the target, and followed by large elements. Since the adversary does not need to report the numerical values, the actual numbers can be chosen with the instantiation in such a way that the latter interval is sorted.

The lower bound is thus obtained by plugging in $k=\krbswap$ into the lower bound of
\[
2\sqrt{4k^2+2kn}-4k=2\sqrt{2kn}-o(\sqrt{kn}).
\]
For $k=\krbswap$ this yields the claimed lower bound.
\end{proof}

A lower bound of $2\sqrt{2n\kbswap}-o(\sqrt{n\kbswap})$ now follows from $\kbswap\leq \krbswap$ (Proposition~\ref{prop:kbswap_krbswap}), but a better lower bound is obtained in the following theorem. 

\begin{theorem}
Every search algorithm needs at least $4\sqrt{n\kbswap}-o(\sqrt{n\kbswap})$ queries to find $e$ even if $\pos(e)=\rank(e)$.\footnoteref{foot:trivial_bound}
\label{thm:LB_kbswap}
\end{theorem}

\begin{proof}[Proof sketch.]
Again we use the same adversary setup as in Theorem~\ref{theorem:app:kaiv:lower} and this time focus on (unrestricted) block swaps rather than restricted block swaps. To get a stronger lower bound, the adversary will use more hidden blocks, namely $2\kbswap-1$ hidden blocks. This of course requires using less than one block swap for each hidden block (unlike the previous case of restricted block swaps). The following claim shows an appropriate construction of $k$ hidden blocks with only $\frac{k+1}{2}$ block swaps. It is formulated for the case of placing $e$ in the second half of an array but the other case is symmetric. Positions are numbered $1$ through $n$.

\begin{claim}
Let $n\in\N$, and let $\alpha_1,\ldots,\alpha_k,\beta_0,\beta_1,\ldots,\beta_k\in\N$ with $\sum \alpha_i=q$ and $\sum \beta_i=\frac{n}{2}-q-1$. There is an array $A$ with $\kbswap(A)\leq \lceil\frac{k+1}{2}\rceil$ that contains (in order) the following elements: $\frac{n}{2}$ elements smaller than $e$, $q$ elements larger than $e$, element $e$, and an alternating sequence of blocks of larger and smaller elements of sizes $\beta_0, \alpha_1, \beta_1, \ldots, \alpha_k, \beta_k$ (with $\beta_i$ the sizes of blocks of larger elements).
\end{claim}

\begin{proof}
Start with any sorted array $A'$ that contains element $e$ in position $\frac{n}{2}+q+1$. Accordingly, with positions numbered $1$ to $n$, array $A'$ has exactly $\frac{n}{2}+q=\frac{n}{2}+\sum \alpha_i$ elements that are smaller than $e$ and exactly $\frac{n}{2}-q-1=\sum \beta_i$ elements that are larger than $e$. We will construct the desired array $A$ from $A'$ by a sequence of at most $\lceil\frac{k+1}{2}\rceil$ block swaps; we start with $A:=A'$.

As a first block swap, exchange $A[\frac{n}{2},\frac{n}{2}-q-1]$ with $A[\frac{n}{2}+q+\beta_0,\frac{n}{2}+q+\beta_0+q]$. In $A$ we now have the following structure: $\frac{n}{2}$ smaller elements, $q$ larger elements, element $e$, $\beta_0$ larger elements, $q$ smaller elements, and $\beta_1+\ldots+\beta_k$ larger elements. All further operations will only be among these final two groups of elements. For convenience, we discuss the remaining operations on the subarray $\hat{A}$ containing only the final $q=\alpha_1+\ldots+\alpha_k$ smaller elements followed by $\beta_1+\ldots+\beta_k$ larger elements. Clearly, operations turning $\hat{A}$ into an alternating sequence of smaller and larger blocks of sizes $\alpha_1,\beta_1,\alpha_2,\ldots,\alpha_k,\beta_k$ can also be applied to get $A$ into the required form. We show how to do this with $\lceil\frac{k-1}{2}\rceil$ block swaps.

If $k=1$ then $\hat{A}$ has already the required form and we use $0=\lceil\frac{k-1}{2}\rceil$ block swaps; getting $A$ into correct form thus used $1=\lceil\frac{k+1}{2}\rceil$ block swaps. If $k=2$ then we need to transform $\alpha_1+\alpha_2$ small elements followed by $\beta_1+\beta_2$ large ones into pattern $\alpha_1,\beta_1,\alpha_2,\beta_2$, which can be done by swapping the last $\alpha_2$ small elements with the first $\beta_1$ large ones; in total we use two swaps on $A$. 

For $k\geq 3$ we show how one block swap reduces the remaining subproblem to one with $k'=k-2$. We have $\alpha_1+\ldots+\alpha_k$ small elements followed by $\beta_1+\ldots+\beta_k$ large ones, and need to reach pattern $\alpha_1,\beta_1,\alpha_2,\ldots,\alpha_k,\beta_k$. We will swap $\hat{A}[\alpha_1+1,\ldots,\alpha_1+\alpha_k]$ with $\hat{A}[|\hat{A}|-\beta_k-\beta_1,|\hat{A}|-\beta_k-1]$, i.e., we are swapping $\alpha_k$ small elements with $\beta_1$ small ones. The result is that $\hat{A}$ now contains (in order) the following blocks: 
(1) $\alpha_1$ small elements (not moved this time), (2) $\beta_1$ large elements (just swapped), (3) $\alpha_2+\ldots+\alpha_{k-1}$ small elements (not swapped this time, but possibly shifted)\footnote{The shifting of elements will not be relevant here but we mention it once to point out that it is not overlooked.}, (4) $\beta_2+\ldots+\beta_{k-1}$ large elements (not swapped), (5) $\alpha_k$ small elements (just swapped), and (6) $\beta_k$ large elements (never moved).

Observe that the remaining problem now becomes to transform a subarray $\tilde{A}$ with $\alpha_2+\ldots+\alpha_{k-1}$ small elements followed by $\beta_2+\ldots+\beta_{k-1}$ large ones into one with pattern $\alpha_2,\beta_2,\alpha_3,\ldots,\alpha_{k-1},\beta_{k-1}$. This part is situated in $\hat{A}[\alpha_1+\beta_1+1,|\hat{A}|-\alpha_k-\beta_k-1]$ and hence also in $A$, and performing the block swaps on this part does not affect the already correctly placed elements. Thus, overall we need at most $\lceil\frac{k+1}{2}\rceil$ block swaps, as claimed.
\end{proof}

Thus, for a lower bound in terms of the number $\kbswap$ of block swaps the adversary can use $k=2(\kbswap-1)$ hidden blocks: An array instantiation with the $k$ hidden blocks in the required positions costs him only $\lceil\frac{k+1}{2}\rceil=\kbswap$ block swaps. Using this, we can plug $k=2(\kbswap-1)$ into the lower bound of
\[
2\sqrt{4k^2+2kn}-4k=2\sqrt{2kn}-o(\sqrt{kn})
\]
and obtain the claimed lower bound of $4\sqrt{n\kbswap}-o(\sqrt{n\kbswap})$.
\end{proof}

Now, we directly get a lower bound in terms of the number $\kbmov$ of block moves since, using $\kbmov\leq 2\kbswap$, a more efficient search would otherwise violate the lower bound for $\kbswap$.

\begin{corollary}\label{corollary:kbm:lower}
Every search algorithm needs at least $2\sqrt{2n\kbmov}-o(\sqrt{n\kbmov})$ queries to find $e$ even if $\pos(e)=\rank(e)$.\footnoteref{foot:trivial_bound}
\label{cor:LB_kbmov}
\end{corollary}

A matching upper bound for $\kbswap$, in the sense of $4\sqrt{n\kbswap}+o(\sqrt{n\kbswap})$, follows immediately from the fact that $\kainv\leq 2\kbswap$ (Proposition~\ref{prop:kainv_kbswap}). The same bound can be obtained relative to the number of block moves, using $\kainv\leq2\kbmov$ or $\kbswap\leq \kbmov$, but a tight upper bound of $2\sqrt{2n\kbmov}+o(\sqrt{n\kbmov})$ is proved in Theorem~\ref{theorem:kbm:upper} below. For the number of restricted block swaps, i.e., swapping only blocks of the same size (never incurring any shifts), we get an upper bound of $4\sqrt{n\krbswap}+o(\sqrt{n\krbswap})$ from $\kbswap\leq\krbswap$ (Proposition~\ref{prop:kbswap_krbswap}) but this is not tight regarding the leading constant but asymptotically tight.

\begin{corollary}
We can find $e$ using $4\sqrt{n\kbswap}+o(\sqrt{n\kbswap})$ queries if $\pos(e)=\rank(e)$.
\label{cor:UB_kbswap}
\end{corollary}

\begin{corollary}
We can find $e$ using $4\sqrt{n\krbswap}+o(\sqrt{n\krbswap})$ queries if $\pos(e)=\rank(e)$.
\label{cor:UB_krbswap}
\end{corollary}

\begin{theorem}\label{theorem:kbm:upper}
We can find $e$ using $2\sqrt{2n\kbmov}+o(\sqrt{n\kbmov})$ queries if $\pos(e)=\rank(e)$.
\label{thm:UB_kbmov}
\end{theorem}

\begin{proof}
The idea for the proof is to revisit the upper bound for number $\kainv$ of adjacent inversions and observe that only the partition into blocks of elements smaller or greater than $e$ matter (in addition to the single element $e$). Such a partition can have at most $2\kainv+2$ blocks because every block of large elements is followed by an adjacent inversion. We can get a similar bound in terms of the number $\kbmov$ of block moves, which allows us to conclude the analog bound of $2\sqrt{2n\kbmov}+o(\sqrt{n\kbmov})$ for $\kbmov$. (Note that any sequence of $k$ block moves that sorts the given array, can be reversed into one turning a sorted array into the given one.)

\begin{claim}
Starting from a sorted array containing at most a single copy of $e$, any sequence of at most $k$ block moves gives a partition into at most $2k+2$ maximal blocks of elements smaller, respectively larger than $e$, and possibly a unit block for the target.
\end{claim}

\begin{proof}
There is nothing to prove if there are only small or only large elements (but we give the proof independent of presence of $e$). For convenience, assume that $A$ only contains elements $x$, $y$, and $e$ with $x<e$ and $y>e$ (i.e. possibly many copies of $x$ and $y$, and at most a single copy of $e$). We consider a single block move and show that it increases the number of maximal $x$- or $y$-blocks by at most $2$. Concretely, we show that the number of alternations between blocks increases by at most two (where we also count alternations between $x$- or $y$-blocks with $e$).

Say that we have $A=(\ldots,a,b,\ldots,c,d,\ldots,e,f,\ldots)$ and we move block $(d,\ldots,e)$ between $a$ and $b$, obtaining $A'=(\ldots,a,d,\ldots,e,b,\ldots,c,f,\ldots)$. Note that $a,b,c,d,e,f\in\{x,y,e\}$. We claim that this block move cannot increase the number of block alternations by three. Assume, for contradiction, that it indeed increases this number by three. Since there are only three new adjacencies, it follows that $(a,d)$, $(e,b)$, and $(c,f)$ must be alternations, i.e., $a\neq d$, $e\neq b$, and $c\neq f$. Similarly, we may not have removed alternations (or else the increase is at most two) so $a=b$, $c=d$, and $e=f$. It follows immediately that none of $a,\ldots,f$ is equal to $e$ since that would imply having at least two copies. Accordingly, $a,b,c,d,e,f\in\{x,y\}$ and we will use $\overline{x}=y$ and $\overline{y}=x$, which allows us to replace, e.g., $a\neq d$ by $a=\overline{d}$. Thus, we get six equalities that together yield $a=\overline{d}=\overline{c}=f=e=\overline{b}=\overline{a}$; a contradiction. It follows that no block move can create more than two additional block alternations, i.e., no block move can increase the number of maximal blocks by more than two.

If $e$ is present then the sorted array has small elements, followed by $e$, followed by large elements; a total of three blocks. This increases to at most $2k+3$ blocks after $k$ block moves, one of which is $e$. If $e$ is not present then we go from $2$ blocks to at most $2k+2$. This completes the proof of the claim.
\end{proof}

It follows that arrays that can be sorted with at most $\kbmov$ block moves have at most $\kbmov+1$ pairs consisting of a large element followed by a small element: Such pairs can only occur between different blocks, neither of which is the block containing $e$. There are at most $2\kbmov+2$ other blocks and hence at most $2\kbmov+1$ alternations between such blocks. Clearly, only every second block alternation can be from larger to smaller element, giving the claimed number of at most $\kbmov+1$ adjacent alternations between an element larger than $e$ and an element smaller than $e$. Alternations of this type are the defining quantity for the algorithm given in Theorem~\ref{theorem:app:kaiv:upper} for parameter $\kainv$; there we used that this number is at most $\kainv$ since they are a special case of adjacent alternations. This yields the claimed upper bound of $2\sqrt{2n\kbmov}+o(\sqrt{n\kbmov})$, completing the proof.
\end{proof}


\section{Conclusion}

We presented upper and lower bounds for the worst-case query complexity of comparison-based search algorithms that are robust to persistent and temporary read errors, or are adaptive to partially disordered input arrays. 
For many cases we gave algorithms that are optimal up to lower order terms.
In addition, many of the algorithms are oblivious to the value of the parameter quantifying errors/disorder, assuming the target element is present in the array. 
In most cases, for small values of~$k$, the dependence of our algorithms on the number $n$ of elements is close to~$\log n$, with only additive dependency on the number of imprecisions. 
In other words, these results smoothly interpolate beween parameter regimes where algorithms are as good as binary search and the unavoidable worst-case where linear search is best possible.

That said, why should one be interested in, e.g., almost tight bounds relative to the number of block moves that take $A$ to a sorted array, as the bounds are far from binary search? 
The point is that only the total number of comparisons matter, and having a worse function that depends on a (in this case) much smaller parameter value can be favorable to having a much better function of a large parameter value.
E.g., after a constant number of block swaps the parameters $\kmax$, $\ksum$ etc.\ may have value $\Omega(n)$ and the guaranteed bound becomes trivial, while running the search algorithm for the case of few block swaps guarantees $\Oh(\sqrt{n})$ comparisons. Similarly, having tight bounds for the various parameters gives us the exact (worst-case) regime for the chosen parameter (in terms of $n$) where a sophisticated algorithm can outperform linear search, or even be as good as binary search.

Despite having already asymptotic tightness, it would be interesting to close the gaps between coefficients of dominant terms in upper and lower bounds for some of the cases. 
Another question would be to find a different restriction than $\pos(e)=\rank(e)$, i.e., the target being in the correct position relative to sorted order, that avoids degenerate lower bounds of $\Omega(n)$ queries for several parameters. 
A relaxation to allowing a target displacement of $\ell$ and giving cost in terms of $n$, $k$, and $\ell$ seems doable in most cases, but is unlikely to be particularly insightful.
Finally, it seems interesting to study whether randomization could lead to improved algorithms for some of the cases. 
The analysis of randomized lower bounds requires entirely new adversarial strategies since the adversary must choose an instantiation without access to the random bits of the algorithm.

\subparagraph*{Acknowledgements.}
The authors are grateful to several reviewers for their helpful remarks regarding presentation and pertinent literature references.


\bibliographystyle{plain}
\bibliography{robustsearch}


\newpage
\appendix


\section{Relations between measures of array disorder}\label{section:relations:disorder}
\label{sec:relations}


\subsection{Equivalences}

\begin{proposition}\label{prop:kseq_kmov}\label{prop:kseq_krep}
$\kseq=\kmov=\krep$.
\end{proposition}

\begin{proof}
$\kseq\leq\kmov$: If $\kmov$ element moves lead to a sorted array then the at least $n-\kmov$ elements that are not moved must form a sorted subsequence.

$\kmov\leq\kseq$: The $\kseq$ elements that are not part of any fixed ordered subsequence of length $n-\kseq$ can be moved to the correct positions in that sequence by $\kseq$ element moves.

$\kseq\leq\krep$: Assume that $\krep$ element replacements suffice to reach a sorted array. It follows that at least $n-\krep$ elements are not replaced and their subsequence must be already sorted.

$\krep\leq \kseq$: Let $P$ a set of $\kseq$ positions such that the subsequence $S$ on the remaining $n-\kseq$ elements is sorted. One can replace the elements in $P$ such that the entire array is sorted.
\end{proof}


\begin{proposition}\label{prop:kaswap_kinv}
$\kaswap=\kinv$.
\end{proposition}

\begin{proof}
If there are any inversions then there is an adjacent inversion, which can be removed by a single swap of the adjacent elements; this lowers the number of inversions by one. On the other hand, each swap of adjacent elements affects only the relative ordering of these elements and, hence, removes at most one inversion.
\end{proof}


\subsection{Relations by similarity of operations}

The following relations hold because some number of operations in terms of the first measure can be used to implement one operation of the second measure.

\begin{proposition}\label{prop:kswap_kaswap}
$\kswap\leq\kaswap$.
\end{proposition}


\begin{proposition}\label{prop:krbswap_kswap}
$\krbswap\leq\kswap$.
\end{proposition}


\begin{proposition}\label{prop:kbswap_krbswap}
$\kbswap\leq\krbswap$. 
\end{proposition}


\begin{proposition}\label{prop:kbmov_kmov}
$\kbmov\leq\kmov$. 
\end{proposition}


\begin{proposition}\label{prop:kbswap_kbmov}
$\kbswap\leq\kbmov\leq2\kbswap$.
\end{proposition}

\begin{proof}
Any block move can be implemented by swapping the block with an empty block. Any block swap can be implemented by two block moves.
\end{proof}


\begin{proposition}\label{prop:krep_kswap}
$\krep\leq2\kswap$.
\end{proposition}

\begin{proof}
A swap of two positions in the array can be implemented by two replacements.
\end{proof}


\subsection{Further relations}

\begin{proposition}\label{prop:ksum_kinv}
$\kinv \leq \ksum \leq 2\kinv$.
\end{proposition}

\begin{proof}
$\kinv\leq\ksum$: If at least one element is displaced then at least one element $e$ occurs before $\rank(e)$ and at least one element $e'$ occurs after $\rank(e')$. If any element $e$ occurs before $\rank(e)$ then there must be a subsequent element $e'$, i.e., with $\pos(e)<\pos(e')$, that occurs after $\rank(e')$. Consider a pair $e$ and $e'$ of elements with $\pos(e)<\rank(e)$ and $\pos(e')>\rank(e')$ such that all elements $e''$ between $\pos(e)$ and $\pos(e')$ have their correct position $\pos(e'')=\rank(e'')$. (Possibly there are no such elements $e''$, but it should be clear that $e_1$ and $e_2$ can always be found if $\ksum>0$.) It follows that $\rank(e)\geq \pos(e')$ since $\rank(e)>\pos(e)$ and the positions between $\pos(e)$ and $\pos(e')$ are already filled with non-displaced elements. Similarly, $\rank(e')\leq\pos(e)$. Consider the operation of swapping $e$ and $e'$: This would lower the total displacement by $2(\pos(e')-\pos(e))$ since $\rank(e')\leq\pos(e)<\pos(e')$ and $\rank(e)\geq\pos(e')>\pos(e)$. In terms of swaps of adjacent elements the swap of $e$ and $e'$ costs exactly $2(\pos(e')-\pos(e)-1$. The relation follows since the number of adjacent swaps is at most the decrease in terms of total displacement.

$\ksum\leq2\kinv$: Recall that $\kinv=\kaswap$. Swapping any two adjacent elements can lower the total displacement by at most two since the elements are moved a total of two positions.
\end{proof}

\begin{proposition}\label{prop:kmax_kinv}
$\kmax\leq\kinv$.
\end{proposition}

\begin{proof}
Recall that $\kinv=\kaswap$. Every swap of adjacent elements moves the elements by exactly one position each. Thus, the maximum displacement is lowered by at most one.
\end{proof}


\begin{proposition}\label{prop:krep_kinv}
$\krep\leq\kinv$. 
\end{proposition}

\begin{proof}
Consider the first element, say $x$, in the sequence that is in an inversion (with some later element). It follows that elements preceding $x$ are not larger than any later element. In particular, the directly preceding element, say $y$, must be strictly smaller than $x$ and not exceed any later element. Now, if some later element is equal to $y$ then replace $x$ by $y$; else, replace it by an arbitrary value that is larger than $y$ but smaller than any later element. Clearly, in both cases all inversions involving $x$ are handled (at least one), proving the bound. (Note that the replacement rules ensure that arrays of unique numbers will retain this property. Setting $x$ to the value of $y$ is only done if needed, i.e., if that value already occurs at least one more time.)
\end{proof}


\begin{proposition}\label{prop:kainv_kbswap}
$\kainv\leq2\kbswap$. 
\end{proposition}

\begin{proof}
This can be proved by analyzing the three different types of block swaps: (i) block moves, i.e., swapping a nontrivial block with an empty block, (ii) swapping two nontrivial, nonadjacent blocks, (iii) swapping two nontrivial, adjacent blocks. Cases (i) and (iii) can be verified to only increase the number of adjacent inversions by at most two, else leading to a simple contradiction.
For (ii), assume that we start with
\[
\ldots,a,b,\ldots,c,d,\ldots,a',b',\ldots,c',d',\ldots
\]
and swap $b,\ldots,c$ with $b',\ldots,c'$ to obtain
\[
\ldots,a,b',\ldots,c',d,\ldots,a',b,\ldots,c,d',\ldots.
\]
Note that only the eight pairs $(a,b)$, $(c,d)$, $(a',b')$, $(c',d')$, $(a,b')$, $(c',d)$, $(a',b)$, and $(c,d')$ matter for upper-bounding the increase in number of adjacent inversions. If $a>b'$ and $a'>b$ then $a>b$ or $a'>b'$ must hold (depending on $b\geq b'$ or $b<b'$); in other words, if there are adjacent inversions at both $(a,b')$ and $(a',b)$ after the block swap then among $(a,b)$ and $(a',b')$ there was at least one adjacent inversion. Similarly, if $c'>d$ and $c>d'$ then $c>d$ or $c'>d'$ must hold, i.e., if we have inversions at both $(c',d)$ and $(c,d')$ then we had at least one adjacent inversion among $(c,d)$ and $(c',d')$. Thus, the total of adjacent inversions increases by at most two with this block swap.
\end{proof}


\begin{proposition}\label{prop:kainv_kseq}
$\kainv\leq\kseq$ 
\end{proposition}

\begin{proof}
Pick any $\kseq$ positions $P$ such that the subsequence $S$ of the remaining $n-\kseq$ positions is sorted. Clearly, any adjacent inversions must be between elements of $P$ or between an element of $P$ and an element of $S$. Consider any subsequence $x,y_1,\ldots,y_p,z$ where $p\geq 1$ and $y_1,\ldots,y_p\in P$ and $x,z\in S$. If there is the maximum of $p+1$ adjacent inversions then it follows that $x>y_1>\ldots>y_p>z$, violating that $x<z$ in the sorted subsequence $S$. Else, there are at most $p$ adjacent inversions incident with $y_1,\ldots,y_p\in P$. Thus, overall, have at most $\kseq$ adjacent inversions.
\end{proof}


\subsection{Unboundedness results}

Here we give pairs of measures such that the second can be unbounded, even if the first is constant. Each such relation is represented by a dashed red arc in Figure~\ref{fig:relations}. Note that for any pair of parameters not connected with a directed path in the figure, unboundedness follows, because any bound would produce some path that contradicts an unboundedness relation.

The first proposition of this type covers the comparison of all parameters other than $\kmax$ with $\kmax$ since we showed that $\kainv$ is bounded whenever any other parameter (except $\kmax$) is bounded.

\begin{proposition}
There exist arrays with $\kmax=1$ and $\kainv=\Omega(n)$.
\end{proposition}

\begin{proof}
For given even integer $n$, consider the array $A=[2,1,4,3,\ldots,n,n-1]$. Clearly, each element $e$ is (exactly) one position away from $\rank(e)$. However, we find that the array has $\frac{n}{2}=\Omega(n)$ adjacent inversions.
\end{proof}


\begin{proposition}
There exist arrays with $\kswap=1$ and $\kmax=\Omega(n)$.
\end{proposition}

\begin{proof}
For given $n$, consider the array $A=[n,2,3,\ldots,n-2,n-1,1]$. A single element swap suffices to reach a sorted array, but the maximum displacement is $n-1=\Omega(n)$. 
\end{proof}


\begin{proposition}
There exist arrays with $\krep=1$ and $\kmax=\Omega(n)$.
\end{proposition}

\begin{proof}
For given $n$, consider the array $A=[n+1,2,3,\ldots,n-2,n-1,n]$. A single element replacement, namely $n+1$ by $1$, suffices to reach a sorted array, but the maximum displacement is $n-1=\Omega(n)$. 
\end{proof}


\begin{proposition}
There exist arrays with $\krep=1$ and $\krbswap=\Omega(\log n)$.
\end{proposition}

\begin{proof}
We prove by induction on~$d$ that an array~$A$ of size~$n = 4^d$ with a large element followed by an increasing sequence of smaller elements cannot be sorted with less than~$d$ restricted block swaps.
The claim trivially holds for~$d\in\{0,1\}$.
For~$d>1$, consider any shortest possible sequence of restricted block swaps to sort the array.
Let~$A'$ be the array after executing the first swap only and let $i$ be the index of~$A[1]$ in~$A'$.
If~$i \leq 3n/4$, we can apply induction on~$A'[i, i+1, \dots, i+n/4-1]$.
If~$i > 3n/4$, let~$j$ be the largest index with~$A[j] \neq A'[j]$ and let~$j'$ be the index of~$A[j]$ in~$A'$.
We can then apply induction on~$A'[j', j'+1, \dots, j'+n/4-1]$ to obtain the claimed bound.
\end{proof}


\begin{proposition}
There exist arrays with $\krbswap=1$ and $\kseq=\Omega(n)$.
\end{proposition}

\begin{proof}
For given even integer $n$, consider the array $A=[\frac{n}{2},\frac{n}{2}+1,\ldots,n,1,2,\ldots,\frac{n}{2}-1]$. Clearly, $A$ can be turned into a sorted array by a single restricted block swap, i.e., it has $\krbswap=1$. Its longest sorted subsequence, however, has length $\frac{n}{2}$, implying $\kseq=\frac{n}{2}=\Omega(n)$.
\end{proof}


\begin{proposition}\label{prop:kbswap_kainv}
There exist arrays with $\kainv=1$ and $\kbswap=\Omega(n)$.
\end{proposition}

\begin{proof}
For a given even integer $n$, consider the array $A=[1,3,5,\dots,n-1, 2, 4, \dots, n]$.
Clearly, $A$ has a single adjacent inversion.
To sort the array with block swaps, there needs to be a block swap that increases the distance between consecutive odd/even numbers.
Each block swap can increase at most one such distance, hence we need at least~$\frac{n}{2} - 1$ block swaps to sort the array.
\end{proof}


\end{document}